\documentclass[11pt]{amsart}
\usepackage{graphicx,amssymb,amsmath,amsthm}
\usepackage{enumerate}
\usepackage{comment,cite,color}
\usepackage{cite,color}
\usepackage{mathrsfs}
\usepackage{epsfig}
\usepackage{lscape}
\usepackage{subfigure}
\usepackage{epstopdf}
\usepackage{array}
\usepackage{caption}
\usepackage{algorithm}
\usepackage{algpseudocode}

\newcommand{\innerp}[1]{\langle {#1} \rangle}

\newcommand{\abs}[1]{\lvert#1\rvert}

\renewcommand{\Re}{{\rm Re}}

\newcommand{\R}{{\mathbb R}}

\newcommand{\C}{{\mathbb C}}

\renewcommand{\eqref}[1]{(\ref{#1})}

\newcommand{\conj}[1]{\overline{#1}}

\newcommand{\E}{{\mathbb E}}

\newcommand{\vx}{{\mathbf x}}
\newcommand{\vy}{{\mathbf y}}

\newcommand{\vs}{{\mathbf s}}
\newcommand{\va}{{\mathbf a}}
\newcommand{\vz}{{\mathbf z}}
\newcommand{\vb}{{\mathbf b}}

\newcommand{\vv}{{\mathbf v}}
\newcommand{\vw}{{\mathbf w}}
\newcommand{\vu}{{\mathbf u}}

\newcommand{\vg}{{\mathbf g}}
\newcommand{\ve}{{\mathbf e}}
\newcommand{\vf}{{\mathbf f}}
\renewcommand{\H}{{\mathbb H}}

\newcommand{\MM}{\mathbf M}

\newcommand{\ud}{{\,\mathrm d}}

\newtheorem{definition}{Definition}[section]

\newtheorem{theorem}{Theorem}[section]

\newtheorem{lemma}{Lemma}[section]
\newtheorem{remark}{Remark}[section]

\newcommand{\zz}{^{\top}}

\usepackage[top=1.2in,bottom=1.2in,left=1.2in,right=1.2in]{geometry}
\date{}

\title{Affine Phase Retrieval via Second-Order Methods}

\author{Bing Gao}
\address{School of Mathematical Sciences and LPMC, Nankai University, Tianjin 300071, P.R. China}
\email{gaobing@nankai.edu.cn}
\thanks{This work was funded by NSFC grant 12001297.}

\begin{document}
	
	
	\maketitle
	\begin{abstract}		
		In this paper, we study the affine phase retrieval problem, which aims to recover signals from the magnitudes of affine measurements. We develop second-order optimization methods based on Newton and Gauss-Newton iterations and establish that, under specific a priori conditions, the problem exhibits strong convexity. Theoretically, we prove that the Newton method with resampling achieves global quadratic convergence in the noiseless setting for both Gaussian measurements and admissible coded diffraction patterns (CDPs). Furthermore, we demonstrate that the same theoretical framework naturally extends to the Gauss-Newton method, implying its quadratic convergence. To validate our theoretical findings, we conduct extensive numerical experiments. The results confirm the quadratic convergence of second-order methods, while their computational efficiency remains comparable to that of first-order methods. Additionally, our experiments demonstrate that second-order methods achieve exact recovery with relatively few measurements, highlighting their practical feasibility and robustness.		
	\end{abstract}
	
	%
	\vspace{2pc}             
	%
	%
	\maketitle
	\markboth{}{}
	%
	%
	\section{Introduction}
	
	\subsection{Phase Retrieval and Affine Phase Retrieval}
	Phase retrieval addresses the problem of recovering a signal or an image from its phaseless measurements. This problem arises in various scientific and engineering fields, such as X-ray crystallography \cite{millane1990phase,miao2008extending}, diffraction imaging \cite{bunk2007diffractive}, optics \cite{walther1963question} and astronomy \cite{fienup1987phase}. In these applications, physical limitations prevent direct measurement of the phase, and only the magnitudes of Fourier transforms or other linear transformations of the signal can be observed. Reconstructing the signal from the magnitudes is difficult because the relationship between the magnitude and the signal is non-linear and non-convex.
	Mathematically, the problem considers recovering a vector $\vx\in\H^n$ (where $\H=\C$ or $\R$ from a quadratic system:
	\begin{equation}\label{pr_quad}
		\{y_j=\abs{\innerp{\va_j,\vx}}^2+w_j, j=1,\ldots,m\},
	\end{equation}
	where $\{\va_1,\ldots,\va_m\}\subset \H^n$ are measurement vectors, $ y_j$ are the corresponding observations and $w_j$ represents noise. A key challenge is that multiplying the signal by a unit constant (e.g. $ e^{i\theta}$, $ \theta\in[0,2\pi) $) yields the same observations, meaning that the signal can only be recovered up to a global phase factor. The problem becomes even more complex when the measurements are Fourier transforms, which introduce both trivial and non-trivial ambiguities \cite{bendory2017fourier}. To address these challenges, prior information such as positivity, sparsity \cite{jaganathan2017sparse,jagatap2017fast}, or structured measurements \cite{jaganathan2016stft,candes2015phase} is often incorporated into the problem.
	
	In recent years, numerous algorithms have been developed to solve the classical phase retrieval problem, encompassing both convex and non-convex approaches. Convex methods, such as PhaseLift \cite{candes2013phaselift,candes_phaselift} and PhaseCut\cite{waldspurger2015phasecut}, reformulate the phase retrieval problem as a semidefinite program, offering strong theoretical guarantees. However, these methods tend to be computationally expensive for large-scale problems. In contrast, non-convex methods, including the Gerchberg-Saxton algorithm\cite{Gerchberg-Saxton} and its variants\cite{fienup1982phase,netrapalli2013phase,waldspurger2018phase}, are more computationally efficient but face challenges such as slow convergence and sensitivity to the initial point. Recently, a wide array of non-convex algorithms based on various optimization models  have been developed to address the phase retrieval problem \cite{WF,twf,altermin,taf,gauss_newton}. These algorithms generally follow a two-step process:  first, generate a well-chosen initial estimate, and second, refine this initial estimate using classical optimization techniques. Commonly employed models in this context include:
	\begin{equation*}
		\min_{{\vz}\in\C^n}\frac{1}{2m}\sum_{j=1}^m(\abs{\innerp{\va_j,\vz}}^2-y_j)^2, \quad  \min_{{\vz}\in\C^n}\frac{1}{2m}\sum_{j=1}^m(\abs{\innerp {\va_j,\vz}}-\sqrt{y_j})^2,
	\end{equation*}
	and
	\begin{equation*}
		\min_{{\vz}\in\C^n}\frac{1}{2m}\sum_{j=1}^m\Big(\abs{\innerp {\va_j,\vz}}^2-y_j\log\big(\abs{\innerp {\va_j,\vz}}^2\big)\Big).
	\end{equation*}
	
	Affine phase retrieval \cite{affine_g} is a closely related problem that generalizes the phase retrieval problem by allowing affine transformations in the measurement process. The goal is to recover $ \vx\in\H^n $ from affine quadratic measurements:
	\[
	\big\{y_j=\abs{\innerp{\va_j,\vx}+b_j}^2+w_j, \,j=1,\ldots,m\big\},
	\]
	where $\vb=(b_1,b_2,\ldots,b_m)^\top\in\H^n $ is a known prior vector. When $ \vb=\boldsymbol{0} $ (zero vector), the problem reduces to the classical phase retrieval problem. Thus, affine phase retrieval can be seen as phase retrieval with side information. 
	It arises in various applications, such as holography \cite{liebling2003local, barmherzig2019holographic}, phase retrieval with reference signal \cite{arab2020fourier,hyder2020solving,hyder2019fourier}, and phase retrieval with background information \cite{elser2018benchmark,yuan2019phase}. 
	In holographic phase retrieval, the use of a known reference signal positioned at a certain distance from the unknown signal serves to linearize the measurements. Similarly, in phase retrieval with background information, the known portion of the signal helps to obtain the amplitudes of linear measurements.
	
	Several algorithms have been developed specifically for affine phase retrieval with Fourier measurements \cite{barmherzig2019holographic,hyder2019fourier,arab2020fourier}. For instance, in \cite{barmherzig2019holographic}, the authors propose a referenced deconvolution algorithm, transforming the affine phase retrieval problem into a linear deconvolution task. Additionally, in \cite{hyder2019fourier,arab2020fourier}, side information or reference signals have been used to regularize the classical Fourier phase retrieval problem, effectively converting it into an affine phase retrieval problem. These approaches have demonstrated significant improvements in reconstruction accuracy and stability, underscoring the benefits of incorporating affine information into phase retrieval.
	
	\subsection{Targets and Motivations}
	Clearly, if we define $ \tilde{\va}_j=(\va_j,\,\conj{b_j})^\top \in\H^{n+1}$ and $\tilde{\vx}= (
	\vx,\, 1)^\top \in\H^{n+1}$, the affine quadratic measurements in $ \vx $, i.e., $ \abs{\innerp{\va_j,\vx}+b_j}^2 $, become equivalent to the quadratic measurements in $\tilde{\vx} $, i.e., $ \abs{\innerp{\tilde{\va}_j,\tilde{\vx}}}^2$. This equivalence demonstrates that the affine phase retrieval problem can be reformulated as a standard phase retrieval problem.  Conversely, if certain components of the signal are known in advance, for example, $ \vx=[\vx_1,\vx_2]\zz\in\H^n $, where $ \vx_2\in\H^{n-n_1}$ is known. In this case, the measurements $ |\va_j^*\vx|^2=|\va_j(1:n_1)^*\vx_1+\va_j(n_1+1:n)^*\vx_2|^2 $ can be expressed with $\va_j(n_1+1:n)^*\vx_2$ as a known vector. This implies that the phase retrieval problem can also be reformulated as an affine phase retrieval problem. 
	
	While the two problems are interconvertible, they exhibit fundamental differences. Unlike classical phase retrieval, affine phase retrieval enables exact signal recovery without ambiguity due to the additional structural information provided by the affine measurements. This structural information not only resolves the inherent ambiguities of classical phase retrieval but also preserves properties that can be effectively leveraged in algorithmic design. This paper focuses on exploiting these advantages to develop second-order algorithms for solving the affine phase retrieval problem. The motivation for this approach stems from the fact that general second-order algorithms are often unsuitable for standard phase retrieval in the complex domain due to the problem's non-convexity and sensitivity to initialization, as discussed in \cite{gauss_newton}. In contrast, the affine phase retrieval, with its additional structural constraints, provides a more stable framework for the application of second-order optimization techniques. Specifically, as shown in Lemma \ref{affine_lemma}, the Jacobian matrix for the affine phase retrieval problem is invertible under suitable conditions, ensuring the positive definiteness of the Hessian matrix. This property enables the effective use of second-order methods, which are otherwise challenging to apply in the classical phase retrieval setting.
	
	Given the measurements $\{y_j=\abs{\innerp{\va_j,\vz}+b_j}^2+w_j, j=1,\ldots,m\}$, where  $w_j$ represents noise (with $w_j=0$ in the noiseless case), we build the solving model as:
	\begin{equation}\label{eq:affine_phase}
		\min_{{\vz}\in\C^n}f(\vz)=\min_{{\vz}\in\C^n}\frac{1}{2m}\sum_{j=1}^m\big(\abs{\innerp {\va_j,\vz}+b_j}^2-y_j\big)^2.
	\end{equation}
	Our algorithm is constructed to solve (\ref{eq:affine_phase}), leveraging the unique properties of affine phase retrieval to achieve efficient and accurate signal recovery.
	
	\subsection{Contributions}
	Since the affine phase retrieval problem can be reformulated as a standard phase retrieval problem, the existing algorithms designed for phase retrieval can be adapted to solve it. In particular, when $\vb$ is a Gaussian random vector, the theoretical guarantees of most phase retrieval algorithms can also be extended to affine phase retrieval. In this paper, we highlight the unique advantages of affine phase retrieval by developing second-order algorithms that are effective even in the complex domain, where traditional second-order methods for phase retrieval often fail due to non-convexity. This work represents the first instance of second-order algorithms being proposed for the affine phase retrieval problem. Our contributions are summarized as follows:
	\begin{itemize}
		\item We develop both Newton and Gauss-Newton iterative algorithms to solve the affine phase retrieval problem based on the optimization model (\ref{eq:affine_phase}).
		\item We prove that the Newton algorithm exhibits quadratic convergence to the true signal without ambiguity in complex domain. The same proof technique can be directly extended to establish the quadratic convergence of the Gauss-Newton method.
		\item The proposed algorithms are robust, requiring no specific conditions on the choice of the initial point, which significantly enhances their practical applicability.
		\item We provide a comprehensive convergence analysis under standard measurement models, including Gaussian (Definition \ref{gaussian}) and admissible coded diffraction patterns (Definition \ref{cdp}).
	\end{itemize}
	
	\begin{remark}\label{GN_remark}
		While the theoretical analysis in this paper focuses on the Newton method, the same proof framework can be applied to establish the quadratic convergence of the Gauss-Newton method. This demonstrates the flexibility of our approach and its applicability to a broader class of second-order optimization algorithms within the context of affine phase retrieval.
	\end{remark}
	\subsection{Notations}
	Throughout this paper, signals and measurements are considered in the complex domain.
	We denote positive constants by $ C $, $ c $ and $ \gamma $, along with their indexed versions, noting that their values may vary with each occurrence. Let $ \vx \in \C^n $  represent the exact signal to be recovered, and $ \vz_k \in \C^n $ denote the $ k $-th iteration point. For simplicity and without loss of generality, we assume $ \|\vx\|=1 $. When no subscript is specified,  $ \|\cdot\| $ refers to the Euclidean norm, i.e., $ \|\cdot\|= \|\cdot\|_2$.

	We use $\mathcal{CN}(0, \sigma^2 I_n)$ to denote an $n$-dimensional circularly symmetric complex Gaussian vector with mean zero and covariance matrix $\sigma^2 I_n$, where $I_n$ is the $n\times n$ identity matrix.
	The sampling vectors $ \va_j\in\C^n, \, j=1,\ldots, m $ are assumed to follow either the Gaussian model (Definition \ref{gaussian}) or the admissible coded diffraction patterns (CDPs) model (Definition \ref{cdp}). The vector $ \vb\in\C^n $ can either be fixed or randomly generated. For simplicity, we set $ \vb = (b,b,\ldots,b)\zz $, where $ b\in\C $.  
	For a vector $ \vs\in\C^n $, $ \vs^\top $ denotes its  transpose, $ \vs^* $ its conjugate transpose, and $ \conj{\vs} $ its conjugate.  The line segment connecting $ \vz_k $ and $ \vx $ is denoted by:
	\[
	S_k\,:=\,\{t \vx+(1-t)\vz_k: \,\,0\leq t\leq 1\}.
	\]
	\begin{definition}[Gaussian model]\label{gaussian}
		The sampling vectors follow the Gaussian model if each $ \va_j\in\C^n\overset{i.i.d.}{\sim} \mathcal{CN}(0,I_n/2) $.
	\end{definition}
	\begin{definition}[Admissible coded diffraction patterns (CDPs) model]\label{cdp}
		The sampling vectors follow the admissible coded diffraction patterns (CDPs) model if the observations are given by
		\[
		y_{(l,k)} = \Big|\sum_{t=0}^{n-1}x[t]\conj{d}_l[t]e^{-i2\pi kt/n}\Big|^2, \quad \quad
		\begin{small}
			\begin{aligned}
				1&\leq l\leq L  \\ 0&\leq k\leq n-1
			\end{aligned}
		\end{small}
		\]
		where the patterns $ d_l[t] $ are i.i.d. with each entry sampled from a distribution $ g $ satisfying: $ |g|\leq M $, $ \E[g]=0 $, $ \E[g^2]=0 $ and $ \E[\abs{g}^4]=2\E^2[\abs{g}^2] $.	
	\end{definition}
	In the admissible CDPs model, we assume $ \E[\abs{g}^2]=1 $ and $ |g|<\sqrt{6} $. For a fixed $ l $, this model collects the Discrete Fourier Transform (DFT) of the signal modulated by $ d_l[t] $ $ (t=1,\ldots,n) $. To align with the notation used earlier, we define $\va_j$ as $ \va_{(l,k)}=\va_{(\lceil j/n\rceil,\, (j\,\text{mod}\,n)-1)} = D_l \vf_k $, where $ 1\leq l\leq L $ and $0\leq k\leq n-1 $. Here, $ D_l $ is a diagonal matrix with modulation patterns $ d_l[t] $ $ (t=1,\ldots,n) $ on its diagonal, and $ \vf^*_k $ represents the rows of the DFT matrix. Further details on this model can be found in \cite{WF,candes2015phase}.
	
	\subsection{Organization}
	The remainder of the paper is organized as follows: In Section \ref{sec2}, we develop iterative algorithms for solving (\ref{eq:affine_phase}) based on the Newton and Gauss-Newton methods. Section \ref{sec3} analyzes the convergence of the Newton method with resampling, proving its quadratic convergence rate under both the Gaussian model and the admissible CDPs model. Section \ref{sec4} presents numerical experiments to validate the theoretical results and compares the proposed algorithm with several first-order methods. Finally, the Appendix provides supporting lemmas and proofs for the theoretical analysis.
	
	\section{Newton and Gauss-Newton Iteration}\label{sec2}
	In this section, we develop second-order algorithms to solve the optimization problem (\ref{eq:affine_phase}). We begin by deriving the Newton iteration and then extend the approach to the Gauss-Newton method.
	\subsection{Complex Gradient and Hessian}
	The objective function $f(\vz) $ is a real-valued function over the complex variable $ \vz$, and can be expressed as:
	\[
	f(\vz)=\frac{1}{2m}\sum_{j=1}^m(\conj{\vz}^\top\va_j\va_j^*\vz+b_j\conj{\va}_j^\top\vz+\conj{b}_j\conj{\vz}^\top\va_j+\abs{b_j}^2-y_j)^2=f(\vz,\conj{\vz}).
	\]
	Using Wirtinger derivatives, the complex gradient and Hessian are defined as:
	\[
	\nabla f =\begin{pmatrix}
		\frac{\partial f}{\partial \vz}, \frac{\partial f}{\partial \conj{\vz}}
	\end{pmatrix}^*,\quad\quad \nabla^2 f = \begin{pmatrix}
		\frac{\partial}{\partial \vz}\Big(\frac{\partial f}{\partial \vz}\Big)^* & \frac{\partial}{\partial \conj{\vz}}\Big(\frac{\partial f}{\partial \vz}\Big)^*\vspace{2mm}\\ \frac{\partial}{\partial \vz}\Big(\frac{\partial f}{\partial \conj{\vz}}\Big)^*&\frac{\partial}{\partial \conj{\vz}}\Big(\frac{\partial f}{\partial \conj{\vz}}\Big)^*
	\end{pmatrix},
	\]
	where
	\[
	\frac{\partial f}{\partial \vz}=\frac{\partial f(\vz,\conj{\vz})}{\partial \vz}\Big|_{\conj{\vz}=\text{constant}},\quad \frac{\partial f}{\partial \conj{\vz}}=\frac{\partial f(\vz,\conj{\vz})}{\partial \conj{\vz}}\Big|_{\vz=\text{constant}}.
	\]
	\subsection{Newton and Gauss-Newton Iteration} 
	Let $ \vz_k\in\C^n $ be the current estimate and $ \Delta\vz_k=\vz-\vz_k\in\C^n $ be an increment. Using second-order Taylor approximation, we have:
	\[
	f(\vz)=f(\vz_k+\Delta\vz_k)\approx f(\vz_k)+\big(\nabla f(\vz_k,\conj{\vz_k})\big)^*\begin{pmatrix}
		\Delta\vz_k\\ \Delta\conj{\vz_k}
	\end{pmatrix}+\frac{1}{2}\begin{pmatrix}
		\Delta\vz_k\\ \Delta\conj{\vz_k}
	\end{pmatrix}^*\nabla^2 f(\vz_k,\conj{\vz_k})\begin{pmatrix}
		\Delta\vz_k\\ \Delta\conj{\vz_k}
	\end{pmatrix}.
	\]
	Minimizing this approximation yields the Newton iteration:
		\begin{equation}\label{iteration}
			\begin{pmatrix}
				\vz_{k+1}\\\conj{\vz_{k+1}}
			\end{pmatrix}=\begin{pmatrix}
				\vz_{k}\\\conj{\vz_{k}}
			\end{pmatrix}-\big(\nabla^2 f(\vz_k,\conj{\vz_k})\big)^{-1}\nabla f(\vz_k,\conj{\vz_k}).
		\end{equation}
		For the objective function (\ref{eq:affine_phase}), the gradient and Hessian are explicitly given by:
		\begin{align*}
			\nabla f(\vz_k, \conj{\vz_k})&=\begin{pmatrix}
				(\partial f/\partial \vz)\big|_{\vz=\vz_k} &(\partial f/\partial \conj{\vz})\big|_{\conj{\vz}=\conj{\vz_k}}
			\end{pmatrix}^*\\
			&=\frac{1}{m}\sum_{j=1}^{m}\begin{pmatrix}
				\Big(\abs{\va_j^*\vz_k+b_j}^2-y_j\Big)(\va_j^*\vz_k+b_j)\va_j\vspace{3mm}\\
				\Big(\abs{\va_j^*\vz_k+b_j}^2-y_j\Big)(\vz_k^*\va_j+\conj{b_j})\conj{\va_j}
			\end{pmatrix},
		\end{align*}
		\begin{align*}
			\nabla^2 f(\vz_k, \conj{\vz_k}) &= \begin{pmatrix}
				\frac{\partial}{\partial \vz}\Big(\frac{\partial f}{\partial \vz}\Big)^* & \frac{\partial}{\partial \conj{\vz}}\Big(\frac{\partial f}{\partial \vz}\Big)^*\vspace{2mm}\\ \frac{\partial}{\partial \vz}\Big(\frac{\partial f}{\partial \conj{\vz}}\Big)^*&\frac{\partial}{\partial \conj{\vz}}\Big(\frac{\partial f}{\partial \conj{\vz}}\Big)^*
			\end{pmatrix}\\
			&=\frac{1}{m}\sum_{j=1}^{m}\begin{pmatrix}
				(2\abs{\va_j^*\vz_k+b_j}^2-y_j)\va_j\va_j^*& (\va_j^*\vz_k+b_j)^2\va_j\va_j^\top \vspace{3mm}\\
				(\vz_k^*\va_j+\conj{b_j})^2\conj{\va_j}\va_j^*& (2\abs{\va_j^*\vz_k+b_j}^2-y_j)\conj{\va_j}\va_j^\top
			\end{pmatrix}.
		\end{align*}
		Similarly, the Gauss-Newton iteration is derived by approximating the Hessian as:
		\[
		G(\vz_k,\conj{\vz_k})^*G(\vz_k,\conj{\vz_k})=
		\frac{1}{m}\sum_{j=1}^{m}\begin{pmatrix}
			\abs{\va_j^*\vz_k+b_j}^2\va_j\va_j^*& (\va_j^*\vz_k+b_j)^2\va_j\va_j^\top \vspace{3mm}\\
			(\vz_k^*\va_j+\conj{b_j})^2\conj{\va_j}\va_j^*& \abs{\va_j^*\vz_k+b_j}^2\conj{\va_j}\va_j^\top
		\end{pmatrix},
		\]
		leading to the iteration:
		\begin{equation}\label{iteration_gn1}
			\begin{pmatrix}
				\vz_{k+1}\\\conj{\vz_{k+1}}
			\end{pmatrix}=\begin{pmatrix}
				\vz_{k}\\\conj{\vz_{k}}
			\end{pmatrix}-\big( G(\vz_k,\conj{\vz_k})^*G(\vz_k,\conj{\vz_k})\big)^{-1}\nabla f(\vz_k,\conj{\vz_k}).
		\end{equation}
		\subsection{The Newton and Gauss-Newton Method with Re-sampling}
		To construct the complete algorithm, we first determine the starting point. Unlike classical phase retrieval, where the selection of the initial point is critical, affine phase retrieval allows for greater flexibility. The initial point $\vz_0$ can be randomly chosen, or for simplicity, we can directly set  $ 
		\vz_0=\boldsymbol{0} $.
		
		Starting from $\vz_0$, we iteratively refine the estimate using the iteration scheme defined in (\ref{iteration}) or (\ref{iteration_gn1}). To ensure theoretical guarantees, we require that the measurements used in each iteration are independent of the iteration point. This is achieved by resampling the measurement matrix $ A $ at each step, a technique commonly employed in prior works such as \cite{altermin}, \cite{WF}, and \cite{gauss_newton}. Specifically, we partition the measurements $ \va_j $, $ b_j $ and observations $ y_j $ into $T$ disjoint blocks of equal size, using one block per iteration. The detailed steps are outlined in Algorithm \ref{Alg}.
		
		\begin{algorithm}[H]
			\caption{Newton/Gauss-Newton Method with Re-sampling for Affine Phase Retrieval}\label{Alg}
			\begin{algorithmic}[h]
				\Require
				Measurement matrix: $ A\in\C^{m_r\times n}$, fixed vector $ \vb=(b,b,\ldots,b)\in\C^{m_r} $, observations: $ \vy\in \R^{m_r} $ and the required accuracy $ \epsilon>0 $.
				\begin{enumerate}
					\item [1:] Set $ T =c\log\log\frac{1}{\epsilon} $, where $ c $ is a sufficiently large constant.
					\item[2:] Partition $ \vy $ and the corresponding rows of $ A $ and the vector $ \vb $ into $ T $ disjoint blocks: $ (\vy^{(0)}, A^{(0)},\vb^{(0)}),(\vy^{(1)}, A^{(1)},\vb^{(1)}),\ldots, (\vy^{(T)}, A^{(T)},\vb^{(T)})$. Each block $A^{(j)}$ contains $m =m_r/T$ rows.
					\item[3:] 				
					Choose $ \vz_0 = \boldsymbol{0}$ as the initial point.
					\item[4:] For $ k=0, 1,\ldots, T-1 $, update $\vz_{k+1}$ according to 
					\[
					\begin{pmatrix}
						\vz_{k+1}\\\conj{\vz_{k+1}}
					\end{pmatrix}=\begin{pmatrix}
						\vz_{k}\\\conj{\vz_{k}}
					\end{pmatrix}-\big(\nabla^2 f^{(k)}(\vz_k,\conj{\vz_k})\big)^{-1}\nabla f^{(k)}(\vz_k,\conj{\vz_k})
					\]
					or 
					\[
					\begin{pmatrix}
						\vz_{k+1}\\\conj{\vz_{k+1}}
					\end{pmatrix}=\begin{pmatrix}
						\vz_{k}\\\conj{\vz_{k}}
					\end{pmatrix}-\big( G^{(k)}(\vz_k,\conj{\vz_k})^*G^{(k)}(\vz_k,\conj{\vz_k})\big)^{-1}\nabla f^{(k)}(\vz_k,\conj{\vz_k}),
					\]
					where $ f^{(k)} $ and $ G^{(k)}(\vz_k,\conj{\vz_k})$ are defined using the measurements in $ (\vy^{(k)}, A^{(k)},\vb^{(k)}) $.
					\item[5:] End for 
				\end{enumerate}
				\Ensure
				$\vz_T $.
			\end{algorithmic}
		\end{algorithm}
		
		In Algorithm \ref{Alg}, $ \vb $ is set as a fixed vector with all entries equal to $ b $, where $ b $ is a constant. Alternatively, $ \vb $ could be a random vector with independent entries drawn from a bounded distribution or another fixed known vector. This flexibility allows us to explore various scenarios while maintaining the integrity of the algorithm.
		
		\section{Convergence Property of the Newton/Gauss-Newton Method with Re-sampling}\label{sec3}
		
		In this section, we analyze the convergence of Algorithm \ref{Alg}, focusing on the Newton method. The results extend naturally to the Gauss-Newton method, and to avoid repetition, we provide detailed proofs only for the Newton method.
		
		The convergence analysis relies on two key lemmas: Lemma \ref{lemma_mineig} establishes the strong convexity of the objective function $ f $, while Lemma \ref{Lf} demonstrates the Lipschitz continuity of the Hessian $ \nabla ^2 f $. These lemmas are first presented under the assumption of bounded noise to ensure broader applicability. However, to rigorously prove the global quadratic convergence of the algorithm, we proceed under the assumption of noise-free measurements, as detailed in Theorem \ref{th:maintheorem}.
		\subsection{Two crucial Lemmas}
		
		\begin{lemma}[Strong convexity]\label{lemma_mineig}
			Assume $ \vz_k, \vx\in\C^n $ and let the measurement vectors $ \va_j\in\C^n $, $ j=1,\ldots, m $ be distributed according to either the Gaussian model or the admissible CDPs model. The measurements are independent of both $ \vz_k $ and $ \vx $. Additionally, the noise is bounded by $\|\vw\|_{\infty}\leq \eta \|\vx\|^2$ for some positive constant $\eta>0$.
			
			Let $ \vb=(b,b,\ldots,b)\in\C^m $, where $b $ satisfies  $ |b|^2\geq \max\{4(\|\vx\|^2+\|\vz_k\|^2), 8\eta\|\vx\|^2\} $. For $ 0<\epsilon\leq 2/41 $ and a constant $ \gamma_\epsilon>0$, if the number of measurements $m$ satisfies $ m\geq Cn\log n $ in the Gaussian model or $ L\geq C\log^3n  $ in the admissible CDPs model (where $ C $ is a sufficiently large constant), then  
			\[
			\lambda_{\min}(\nabla^2f(\vz_k,\conj{\vz_k}))\geq \frac{|b|^2}{4}
			\]
			holds with probability at least $ 1-c_1\exp(-\gamma_\epsilon n)-c_2/n^2 $ for the Gaussian model and at least $ 1-c_3(2L+1)/n^3 $ for the admissible CDPs model. Here $ c_1 $, $ c_2 $ and $ c_3 $ are positive constants.
		\end{lemma}
		\begin{proof}
			To simplify notation, let $ S=\nabla^2f(\vz_k,\conj{\vz_k}) $ represent the Hessian matrix in our analysis. We decompose $ S $ into two components: 
			\begin{itemize}
				\item $S_{exa}$, the exact (noise-free) part, and 
				\item $S_{noi}$, the noise part.
			\end{itemize}
			Thus, $S = S_{exa}-S_{noi}$, where 
			\begin{align*}
				S_{exa} := \frac{1}{m}\sum_{j=1}^{m}\begin{pmatrix}
					(2\abs{\va_j^*\vz_k+b_j}^2-\abs{\va_j^*\vx+b_j}^2)\va_j\va_j^*& (\va_j^*\vz_k+b_j)^2\va_j\va_j^\top \vspace{3mm}\\
					(\vz_k^*\va_j+\conj{b_j})^2\conj{\va_j}\va_j^*& (2\abs{\va_j^*\vz_k+b_j}^2-\abs{\va_j^*\vx+b_j}^2)\conj{\va_j}\va_j^\top
				\end{pmatrix}
			\end{align*}
			and 
			\begin{align*}
				S_{noi} :=\frac{1}{m}\sum_{j=1}^{m}\begin{pmatrix}
					w_j\va_j\va_j^*& 0\\
					0& w_j\conj{\va_j}\va_j^\top
				\end{pmatrix}.
			\end{align*}
			From Lemma \ref{WF_lemma1}, for any $\epsilon>0$, the following inequality holds with high probability:
			\[
			\|S_{noi}\|\leq \|w\|_\infty\cdot \Big\|\frac{1}{m}\sum_{j=1}^{m}\va_j\va_j^*\Big\|\leq (1+\epsilon)\eta\|\vx\|^2.
			\]
			Moreover, since $|b|^2\geq 8\eta\|\vx\|^2 $, it follows that 
			\begin{equation}\label{noise_est}
				\|S_{noi}\|\leq \frac{1+\epsilon}{8}|b|^2.
			\end{equation}
			This inequality provides a bound on the influence of noise part in the matrix $S$, allowing us to focus on $S_{exa}$.
			
			Next, we decompose $S_{exa}$ into three distinct parts:
			\[
			S_{exa}=S_1+S_2+S_3,
			\]
			with 
			\[
			S_1 := \frac{1}{m}\sum_{j=1}^{m}\begin{pmatrix}
				\abs{\va_j^*\vz_k+b_j}^2\va_j\va_j^*& 0\\
				0& \abs{\va_j^*\vz_k+b_j}^2\conj{\va_j}\va_j^\top
			\end{pmatrix},
			\]
			\[
			S_2 :=\frac{1}{m}\sum_{j=1}^{m}\begin{pmatrix}
				\abs{\va_j^*\vz_k+b_j}^2\va_j\va_j^*& (\va_j^*\vz_k+b_j)^2\va_j\va_j^\top \vspace{3mm}\\
				(\vz_k^*\va_j+\conj{b_j})^2\conj{\va_j}\va_j^*& \abs{\va_j^*\vz_k+b_j}^2\conj{\va_j}\va_j^\top
			\end{pmatrix},
			\]
			\[
			S_3 :=\frac{1}{m}\sum_{j=1}^{m}\begin{pmatrix}
				-\abs{\va_j^*\vx+b_j}^2\va_j\va_j^*& 0\\
				0& -\abs{\va_j^*\vx+b_j}^2\conj{\va_j}\va_j^\top
			\end{pmatrix}.
			\]
			Since $ \vx $, $ \vz_k $ and the scalars $ b_j= b,\,j=1,\ldots,m $ are all independent of the measurement vectors $ \va_j,\,j=1,\ldots,m $, we apply Lemma \ref{expectations} to compute the expections of $S_i$, $i=1,2,3$:
			\[
			\E[S_1]=\begin{pmatrix}
				(\|\vz_k\|^2+|b|^2)I_n+\vz_k\vz_k^* & 0 \vspace{2mm}\\
				0 & (\|\vz_k\|^2+|b|^2)I_n+\conj{\vz_k}\vz_k^\top
			\end{pmatrix},
			\]
			\[
			\E[S_2]=\begin{pmatrix}
				(\|\vz_k\|^2+|b|^2)I_n+\vz_k\vz_k^* & 2\vz_k\vz_k^\top  \vspace{2mm}\\
				2\conj{\vz_k}\vz_k^* & (\|\vz_k\|^2+|b|^2)I_n+\conj{\vz_k}\vz_k^\top
			\end{pmatrix},
			\]
			and
			\[
			\E[S_3]=\begin{pmatrix}
				-(\|\vx\|^2+|b|^2)I_n-\vx\vx^* & 0 \vspace{2mm}\\
				0 & -(\|\vx\|^2+|b|^2)I_n-\conj{\vx}\vx^\top
			\end{pmatrix}.
			\]
			Summing these, we obtain 
			\begin{align*}
				\E[S_{exa}]&=\E[S_1]+\E[S_2]+\E[S_3]\\
				=&\begin{pmatrix}
					(2\|\vz_k\|^2-\|\vx\|^2+|b|^2)I_n+2\vz_k\vz_k^*-\vx\vx^* & 2\vz_k\vz_k^\top \vspace{2mm}\\
					2\conj{\vz_k}\vz_k^* & (2\|\vz_k\|^2-\|\vx\|^2+|b|^2)I_n+2\conj{\vz_k}\vz_k^\top-\conj{\vx}\vx^\top
				\end{pmatrix}.
			\end{align*}
			According to (\ref{inequality1}) and the conclusions of Lemma \ref{expectation_lemma},
			for any $ 0< \epsilon\leq 2/41 $ and $ \gamma_\epsilon>0$, when $ m\geq C n\log n $  (Gaussian model) or $ L\geq C\log^3 n $ (admissible CDPs model), the following inequalities hold with probability at least $ 1-c_1\exp(-\gamma_\epsilon n)-c_2/n^2 $ (Gaussian model) or $ 1-c_3(2L+1)/n^3 $ (admissible CDPs model). The inequalities are:
			\begin{align*}
				\|S_1-\E[S_1]\|&\leq \frac{\epsilon}{4}\|\E[S_1]\|=\frac{\epsilon}{4}(2\|\vz_k\|^2+|b|^2),\\
				\|S_2-\E[S_2]\|&\leq \frac{\epsilon}{4}\|\E[S_2]\|\leq\frac{\epsilon}{4}(5\|\vz_k\|^2+|b|^2),\\
				\|S_3-\E[S_3]\|&\leq \frac{\epsilon}{4}\|\E[S_3]\|=\frac{\epsilon}{4}(2\|\vx\|^2+|b|^2).
			\end{align*}
			Here the second inequality is according to Lemma \ref{eiganalysis} by setting $ \beta = \|\vz_k\|^2+|b|^2 $, $ \vu=\vv=\vz_k $. These inequalities imply
			\begin{align*}
				\|S_{exa}-\E[S_{exa}]\|&=\|S_1-\E[S_1]+S_2-\E[S_2]+S_3-\E[S_3]\|\\
				&\leq \|S_1-\E[S_1]\|+\|S_2-\E[S_2]\|+\|S_3-\E[S_3]\|\\
				&\leq \frac{\epsilon}{4}\big(7\|\vz_k\|^2+2\|\vx\|^2+3|b|^2\big)\\
				&\leq \frac{\epsilon}{4}\cdot 5|b|^2\leq\frac{5|b|^2}{82}.
			\end{align*}
			The third inequality based on the fact that $ |b|^2\geq 4(\|\vz_k\|^2+\|\vx\|^2) $. Then by Weyl's inequality, we have
			\begin{align}\label{mineig}
				\lambda_{\min}(S_{exa})&\geq \lambda_{\min}(\E[S_{exa}])-\|S_{exa}-\E[S_{exa}]\| \\\nonumber
				&\geq \lambda_{\min}(\E[S_{exa}])-\frac{5|b|^2}{82}.
			\end{align}
			For the Hermitian matrix $ \E[S_{exa}] $, we set $ 2\|\vz_k\|^2-\|\vx\|^2+|b|^2=\beta$, $ \vx=\vu$ and $ \vz_k=\vv $. Then 
			\[
			\E[S_{exa}]= \begin{pmatrix}
				\beta I_n-\vu\vu^* & \boldsymbol{0}\\
				\boldsymbol{0} & \beta I_n-\conj{\vu}\vu^\top
			\end{pmatrix}+2\begin{pmatrix}
				\vv\\ \conj{\vv}
			\end{pmatrix}\begin{pmatrix}
				\vv^*& \vv^\top
			\end{pmatrix}.
			\]
			When $ |b|^2\geq 4(\|\vx\|^2+\|\vz_k\|^2)  $, we have
			\[
			\beta=2\|\vz_k\|^2-\|\vx\|^2+|b|^2>\|\vx\|^2.
			\]	
			Thus, according to Lemma \ref{eiganalysis}, we have
			\begin{align}\label{mineig_es}
				\lambda_{\min}(\E[S_{exa}])& =\beta-\|\vx\|^2=|b|^2+2(\|\vz_k\|^2-\|\vx\|^2)\geq \frac{|b|^2}{2}.
			\end{align}
			Putting (\ref{mineig_es}) into (\ref{mineig}),  we obtain
			\begin{equation}\label{exact_est}
				\lambda_{\min}(S_{exa})\geq \lambda_{\min}(\E[S_{exa}])-\frac{5|b|^2}{82}\geq\frac{18|b|^2}{41}.
			\end{equation}
			Given the constraint $\epsilon\leq 2/41$ in (\ref{noise_est}), we combine
			(\ref{noise_est}) and (\ref{exact_est}) to derive the bound for the full matrix $S$:
			\begin{align*}
				\lambda_{\min}(S)\geq \lambda_{\min}(S_{exa})-\|S_{noi}\|\geq \frac{18|b|^2}{41}-\frac{43|b|^2}{328}\geq \frac{|b|^2}{4}.
			\end{align*}
		\end{proof}
		
		Lemma \ref{lemma_mineig} guarantees that the Hessian is positive definite with a minimum eigenvalue bounded below by $ |b|^2/4 $, provided $|b|^2$ satisfies the signal-dependent lower bound $\max\{4(\|\vx\|^2+\|\vz_k\|^2), 8\eta\|\vx\|^2\}$. This ensures strong convexity  of the objective function in the neighborhood of $ \vz_k $, a necessary condition for the local convergence of the Newton method. 
		\begin{remark}
			Note that the matrix $ G(\vz_k,\conj{\vz_k})^*G(\vz_k,\conj{\vz_k})$  in the Gauss-Newton iteration (\ref{iteration_gn1}) is equivalent to $S_2 $, which defined in the proof.
			Following the same reasoning, we can also conclude that
			\[
			\lambda_{\min}\big(\nabla^2G(\vz_k,\conj{\vz_k})^*G(\vz_k,\conj{\vz_k})\big)\geq \frac{|b|^2}{4}.
			\]
		\end{remark}
		
		To further demonstrate  the quadratic convergence of the algorithm, we now prove the Lipschitz continuity of the Hessian $ \nabla^2 f(\vs, \conj{\vs})$. 
		\begin{lemma}[Lipschitz continuity]\label{Lf}
			Suppose $ \|\vz_k - \vx\|\leq\sqrt{\delta} $, where
			$ \vz_k, \vx\in\C^n $, $\|\vx\|=1$, and	$ \delta \geq 0$ is a constant. Assume the measurement vectors $ \va_j\in\C^n $, $ j=1,\ldots, m $ are distributed according to either the Gaussian model or the admissible CDPs model and are independent of $ \vz_k $ and $ \vx $. Let $\vb = (b_1, b_2, \ldots, b_m) \in \C^m$ be a fixed vector, with each entry equal to $b \in \C$.
			
			For any $ 0<\epsilon\leq 1/24 $ and $ \gamma_\epsilon>0 $, when the number of measurements satisfies $ m\geq Cn\log n $ in the Gaussian model or $ L\geq C\log^3n  $ in the admissible CDPs model for a sufficiently large constant $ C $, the following inequality holds with high probability:
			\[
			\|\nabla^2f(\vs_1,\conj{\vs_1})-\nabla^2f(\vs_2,\conj{\vs_2})\|\leq L_f
			\|\vs_1-\vs_2\|,
			\]
			where
			$ L_f = 12(1+\epsilon)(1+\sqrt{\delta})=25(1+\sqrt{\delta})/2. $
			This result holds with probability at least $ 1-c_1\exp(-\gamma_\epsilon n)-c_2/n^2 $ for the Gaussian model and $ 1-c_3(2L+1)/n^3 $ for the  admissible CDPs model. Here $ c_1, c_2, c_3 $ are positive constants, and $\vs_1,\vs_2 \in S_k=\{t \vx+(1-t)\vz_k: \,\,0\leq t\leq 1\}$.
		\end{lemma}
		\begin{proof}
			To establish the Lipschitz continuity of $\nabla^2f(\vs,\conj{\vs}) $,
			we simplify the notation by setting $ S = \nabla^2f(\vs_1,\conj{\vs_1})-\nabla^2f(\vs_2,\conj{\vs_2}) $. From Lemma \ref{expectations}, we have
			\begin{align}\label{es}
				\E[S]&=2\begin{pmatrix}
					(\|\vs_1\|^2-\|\vs_2\|^2)I_n+\vs_1\vs_1^*-\vs_2\vs_2^* & \vs_1\vs_1^\top-\vs_2\vs_2^\top \vspace{2mm}\\
					\conj{\vs_1}\vs_1^*-\conj{\vs_2}\vs_2^* & (\|\vs_1\|^2-\|\vs_2\|^2)I_n+\conj{\vs_1}\vs_1^\top-\conj{\vs_2}\vs_2^\top
				\end{pmatrix}\\\nonumber
				&=2\Big[(\|\vs_1\|^2-\|\vs_2\|^2)I_{2n}+\begin{pmatrix}
					\vs_1\\\conj{\vs_1}
				\end{pmatrix}\begin{pmatrix}
					\vs_1^*&\vs_1^\top
				\end{pmatrix}-\begin{pmatrix}
					\vs_2\\\conj{\vs_2}
				\end{pmatrix}\begin{pmatrix}
					\vs_2^*&\vs_2^\top
				\end{pmatrix}\Big].	
			\end{align}
			Using the Cauchy-Schwarz inequality, we bound the norm:
			\begin{align}\label{es1}
				\big|\|\vs_1\|^2-\|\vs_2\|^2\big|&=\frac{1}{2}\big|(\vs_1+\vs_2)^*(\vs_1-\vs_2)+(\vs_1-\vs_2)^*(\vs_1+\vs_2)\big|\\\nonumber
				&\leq\|\vs_1+\vs_2\|\|\vs_1-\vs_2\|.
			\end{align}
			Additionally, knowing that $ \Big\|\begin{pmatrix}
				\vs\\\conj{\vs}
			\end{pmatrix}\Big\| = \sqrt{2}\|\vs\|$, we obtain 
			\begin{align}\label{es2}
				&\Big\|\begin{pmatrix}
					\vs_1\\\conj{\vs_1}
				\end{pmatrix}\begin{pmatrix}
					\vs_1^*&\vs_1^\top
				\end{pmatrix}-\begin{pmatrix}
					\vs_2\\\conj{\vs_2}
				\end{pmatrix}\begin{pmatrix}
					\vs_2^*&\vs_2^\top
				\end{pmatrix} \Big\| \\\nonumber
				&= \frac{1}{2}\Big\|\begin{pmatrix}
					\vs_1+\vs_2\\\conj{\vs_1}+\conj{\vs_2}
				\end{pmatrix}\begin{pmatrix}
					\vs_1^*-\vs_2^*&\vs_1^\top-\vs_2^\top
				\end{pmatrix}+\begin{pmatrix}
					\vs_1-\vs_2\\\conj{\vs_1}-\conj{\vs_2}
				\end{pmatrix}\begin{pmatrix}
					\vs_1^*+\vs_2^*&\vs_1^\top+\vs_2^\top
				\end{pmatrix}\Big\|\\\nonumber
				&\leq 2\|\vs_1+\vs_2\|\|\vs_1-\vs_2\|.
			\end{align}
			Thus, inserting (\ref{es1}) and (\ref{es2}) into (\ref{es}), we get
			\begin{equation}\label{esnorm}
				\|\E[S]\|\leq 6(\|\vs_1+\vs_2\|)\|\vs_1-\vs_2\|
			\end{equation} 
			Since $ \vs_1, \vs_2\in S_k:=\{t\vx+(1-t)\vz_k, t\in[0,1]\} $ and $ \|\vz_k-\vx\|\leq \sqrt{\delta} $, it follows that
			\[
			1-\sqrt{\delta}\leq\|\vz_k\|\leq 1+\sqrt{\delta}
			\]  
			and
			\begin{equation}\label{s+s}
				\|\vs_1+\vs_2\|=\|(t_1+t_2)\vx+(2-t_1-t_2)\vz_k\|\leq 2(1+\sqrt{\delta}).
			\end{equation}
			Combining (\ref{s+s}) and (\ref{esnorm}) gives
			\begin{equation}\label{esnorm1}
				\|\E[S]\|\leq 12(1+\sqrt{\delta})\|\vs_1-\vs_2\|.
			\end{equation}
			Applying Lemma \ref{expectation_lemma}, for $ \epsilon=1/24>0 $, $ \gamma_\epsilon>0 $ and a sufficiently large constant $ C $, when $ m\geq Cn\log n$ in the Gaussian model or $ L\geq C\log^3n $ in the admissible CDPs model, we have 
			\[
			\|S-\E[S]\|\leq \epsilon\|\E[S]\|
			\] 
			which leads to
			\[
			\|S\|\leq (1+\epsilon)\|\E[S]\|\leq 12(1+\epsilon)(1+\sqrt{\delta})\|\vs_1-\vs_2\|=\frac{25}{2}(1+\sqrt{\delta})\|\vs_1-\vs_2\|.
			\]
			The result holds  with probability at least $ 1-c_1\exp(-\gamma_\epsilon n)-c_2/n^2 $ for the Gaussian model or $ 1-c_3(2L+1)/n^3 $ for the  admissible CDPs model, where $ c_1, c_2 $ and $ c_3 $ are positive constants.
		\end{proof}
		The Lipschitz continuity of the Hessian guarantees that it varies smoothly along the segment connecting $\vx$  and $\vz_k$, which is essential for maintaining the stability and convergence of the Newton/Gauss-Newton method in the affine phase retrieval problem. The smooth variation of the Hessian prevents large fluctuations between iterations, allowing for a more controlled and predictable update process.
		
		Notably, the Lipschitz continuity property remains unaffected by the presence of noise. While noise may introduce perturbations, it does not disrupt the fundamental smoothness of the Hessian matrix, ensuring the robustness of the algorithm.
		\subsection{The main Theorem}
		Based on the strong convexity of the objective function $f$ (Lemma \ref{lemma_mineig}) and the Lipschitz continuity of the Hessian matrix (Lemma \ref{Lf}), we now establish the quadratic convergence of the algorithm under noiseless measurements.
		\begin{theorem}\label{th:maintheorem}
			Let $\vx\in \C^n$ be an arbitrary vector with $ \|\vx\|=1 $, and assume $ \|\vz_k-\vx\|\leq\sqrt{\delta} $, where $\delta>0$ is a constant. Consider the measurements
			$y_j=\abs{\va_j^*\vx+b_j}^2$, where $\va_j\in \C^n$, $ j=1,\ldots, m $ are distributed according to either the Gaussian model or the admissible CDPs model. The vectors $\va_j$ are independent of both $ \vx $ and $ \vz_{k} $. Let $ \vb=(b,b,\ldots,b)^\top\in\C^m $ be fixed with $ |b|^2\geq \max\{25\sqrt{\delta}(1+\sqrt{\delta}), 4(2+\delta+2\sqrt{\delta})\} $. Let $\vz_{k+1}$  be defined by the update rule (\ref{iteration}). Then, when $m\geq Cn\log n$ (Gaussian model) or $ L\geq C\log^3 n $ (admissible CDPs model), with probability at least
			$ 1-c/n^2 $ or $ 1-cL/n^3 $, we have
			\begin{equation}\label{distance_rela}
				\|\vz_{k+1}-\vx\|\leq\beta\cdot \|\vz_{k}-\vx\|^2\leq\|\vz_{k}-\vx\|,
			\end{equation}
			where
			\begin{equation}\label{eq:beta}
				\beta = (1+\sqrt{\delta})\frac{25}{|b|^2}\leq\frac{1}{\sqrt{\delta}}.
			\end{equation}
		\end{theorem}
		\begin{proof}
			Under noiseless measurements, $ \vx $ is the exact solution to (\ref{eq:affine_phase}), implying $ \nabla f(\vx,\conj{\vx})=\boldsymbol{0} $. Using the iteration step
			\begin{align*}
				\begin{pmatrix}
					\vz_{k+1}\\\conj{\vz_{k+1}}
				\end{pmatrix}=\begin{pmatrix}
					\vz_{k}\\\conj{\vz_{k}}
				\end{pmatrix}-(\nabla^2f(\vz_k,\conj{\vz_k}))^{-1}\nabla f(\vz_k, \conj{\vz_k}),
			\end{align*}
			we can derive the following update rule for the error:
			\begin{align*}
				\begin{pmatrix}
					\vz_{k+1}-\vx\\\conj{\vz_{k+1}}-\conj{\vx}
				\end{pmatrix}&=\begin{pmatrix}
					\vz_{k}-\vx\\\conj{\vz_{k}}-\conj{\vx}
				\end{pmatrix}-(\nabla^2f(\vz_k,\conj{\vz_k}))^{-1}\Big(\nabla f(\vz_k, \conj{\vz_k})-\nabla f(\vx, \conj{\vx})\Big)\\
				&=(\nabla^2f(\vz_k,\conj{\vz_k}))^{-1}\Big[\nabla^2f(\vz_k,\conj{\vz_k})\begin{pmatrix}
					\vz_{k}-\vx\\\conj{\vz_{k}}-\conj{\vx}
				\end{pmatrix}-\Big(\nabla f(\vz_k, \conj{\vz_k})-\nabla f(\vx, \conj{\vx})\Big)\Big]\\
				&=(\nabla^2f(\vz_k,\conj{\vz_k}))^{-1}\Big[\nabla^2f(\vz_k,\conj{\vz_k})\begin{pmatrix}
					\vz_{k}-\vx\\\conj{\vz_{k}}-\conj{\vx}
				\end{pmatrix}-\int_{0}^{1}\nabla^2 f(\vg_k, \conj{\vg_k})\begin{pmatrix}
					\vz_{k}-\vx\\\conj{\vz_{k}}-\conj{\vx}
				\end{pmatrix}\ud t\Big]\\
				&=(\nabla^2f(\vz_k,\conj{\vz_k}))^{-1}\Big[\int_{0}^{1}\big(\nabla^2 f(\vz_k,\conj{\vz_k})-\nabla^2 f(\vg_k, \conj{\vg_k})\big)\begin{pmatrix}
					\vz_{k}-\vx\\\conj{\vz_{k}}-\conj{\vx}
				\end{pmatrix}\ud t\Big].
			\end{align*}
			Here $ \vg_k = \vz_k+t(\vx-\vz_k),\,t\in[0,1]$. 
			Taking the Euclidean norm on both sides, we obtain
			\begin{align}\label{feq}
				\|\vz_{k+1}-\vx\|&\leq
				\|(\nabla^2f(\vz_k,\conj{\vz_k}))^{-1}\|\cdot\int_{0}^{1}\big\|\nabla^2f(\vz_k,\conj{\vz_k})-\nabla^2f(\vg_k,\conj{\vg_k})\big\|\cdot
				\|\vz_{k}-\vx\|\ud t.
			\end{align}
			By Lemma \ref{Lf}, with high probability we have
			\begin{align}\label{eqf}
				&\int_{0}^{1}\big\|\nabla^2f(\vz_k,\conj{\vz_k})-\nabla^2f(\vg_k,\conj{\vg_k})\big\|\ud t\\\nonumber
				&\leq L_f \int_{0}^{1}
				\|\vz_{k}-\vg_k\|\ud t\\\nonumber
				&=L_f 
				\|\vz_{k}-\vx\|\int_{0}^{1}t\ud t= \frac{L_f}{2}
				\|\vz_{k}-\vx\|,
			\end{align}
			where $ L_f=25(1+\sqrt{\delta})/2 $.
			Moreover, as $ |b|^2\geq 4(2+\delta+2\sqrt{\delta}) $, by Lemma \ref{lemma_mineig}, with high probability,
			\begin{equation}\label{finv}
				\|(\nabla^2f(\vz_k,\conj{\vz_k}))^{-1}\|\leq \frac{4}{|b|^2}.
			\end{equation}
			Inserting (\ref{finv}), (\ref{eqf}) into (\ref{feq}) yields:
			\begin{align*}
				\|\vz_{k+1}-\vx\|&\leq \frac{4}{|b|^2}\cdot\frac{L_f}{2}\|\vz_{k}-\vx\|^2\\
				&=(1+\sqrt{\delta})\frac{25}{|b|^2}\|\vz_{k}-\vx\|^2.
			\end{align*}
			Given $ |b|^2\geq 25\sqrt{\delta}(1+\sqrt{\delta}) $, it follows that 
			$
			\beta=25(1+\sqrt{\delta})/|b|^2\leq 1/\sqrt{\delta}.
			$
			Therefore, we have:
			$$ \|\vz_{k+1}-\vx\|\leq\frac{1}{\sqrt{\delta}} \|\vz_{k}-\vx\|^2\leq \|\vz_{k}-\vx\|\leq \sqrt{\delta}. $$
			This shows that the error decreases quadratically, and the convergence results are  maintained as the iteration progresses. 
		\end{proof}
		\begin{remark}
			If the iteration point $\vz_{k+1}$ is generated by Gauss-Newton iteration (\ref{iteration_gn1}), we can also obtain the same convergece result.
		\end{remark}
		Based on this theorem, we now state the overall convergence properties of Algorithm \ref{Alg}, which are summarized in the following result.
		\begin{theorem}
			Let $ \vx\in\C^n $ be an arbitrary vector with $ \|\vx\|=1 $, and let $ \va_j,\,j=1,\ldots,m_r $ be distributed according to either the Gaussian model or the admissible CDPs model. Let $ \vb=(b,b,\ldots,b)^\top $ be a fixed vector with $ |b|^2\geq 50$. Given a target accuracy $ \epsilon>0 $, we have the following results. 
			
			For the Gaussian model, if the number of measurements satisfies
			$ m_r\geq C\log\log\frac{1}{\epsilon}\cdot n\log n $, for the admissible CDPs model, if the number of patterns satisfies $ L_r\geq C\log\log\frac{1}{\epsilon}\cdot \log^3 n $, where  $ C $ is a sufficiently large constant, then with high probability at least $ 1-\tilde{c}/n^2 $ (Gaussian model) or $ 1-\tilde{c}\log^3n/n^3 $ (admissible CDPs model), where $\tilde{c}>0$), Algorithm \ref{Alg} generates an estimate $ \vz_T $ such that
			\[
			\|\vz_T-\vx\|\leq \epsilon.
			\]
		\end{theorem}
		\begin{proof}		
			In Algorithm \ref{Alg}, we set $ T= c\log\log\frac{1}{\epsilon} $ and define $m=m_r/T$ and $ L=L_r/T$. Then, for each measurement set $ A^{(j)}$ in Algorithm \ref{Alg}, the number of measurements satisfies $ m\geq C_1 n\log n $ for the Gaussian model and $ L\geq C_1\log^3 n $ for the admissible CDPs model, where $ C_1 $ is a constant depending on $ C $ and $c$. Starting from the initial point $ \vz_0=\boldsymbol{0}$, we have $ \|\vz_0-\vx\|\leq 1=\sqrt{\delta} $, and let $ |b|^2\geq 50=\max\{25\sqrt{\delta}(1+\sqrt{\delta}), 4(2+\delta+2\sqrt{\delta})\} $. 
			
			By iteratively applying Theorem \ref{th:maintheorem} $ T $ times, we obtain
			\begin{align*}
				\|\vz_T-\vx\|&\leq\beta\|\vz_{T-1}-\vx\|^2\\&\leq \beta^{2^T-1}\|\vz_0-\vx\|^{2^T}\\
				&\leq \epsilon,
			\end{align*}
			which holds with probability at least $ 1-\tilde{c}/n^2 $ (Gaussian model) or $ 1-\tilde{c}L/n^3 $ (admissible CDPs model).
		\end{proof}
		
		This work highlights the key distinctions between affine phase retrieval and the classical phase retrieval problem. In our main theorem, we focus on the noiseless case to rigorously establish the global quadratic convergence of Algorithm \ref{Alg}. While this represents an idealized scenario, it provides a foundation for understanding the algorithm's behavior in practical settings where measurements are subject to noise. When noise is bounded, Lemma \ref{mineig} (strong convexity) remains applicable, ensuring the algorithm's robustness.  
		
		Additionally, the choice of $b$ significantly influences the algorithm's accuracy in the presence of noise, as evidenced by our numerical experiments (Subsection \ref{impactb}). A detailed theoretical analysis of the relationship between $b$ and noise tolerance is left for future work.
		
		
		\section{Numerical Experiments}\label{sec4}
		In this section, we present numerical experiments to validate the theoretical results and demonstrate the effectiveness of the proposed algorithms. All experiments are conducted in the complex domain with $ n=512 $ and measurements are generated according to both the Gaussian model and the admissible CDPs model. Unless otherwise specified, in the following numerical experiments, the ground truth $\vx $ is randomly generated from a complex unit Gaussian distribution $\mathcal{CN}(0, I_n)$, and $\vb$ is a vector with identical entries $50\|\vx\|^2$. For all iterative algorithms, the initial point is set to the zero vector, and the relative error is defined as $\|\vx-\vz\|/\|\vx\|$.  The experiments are divided into three parts:
		\begin{itemize}
			\item Effect of Resampling: Compare Newton and Gauss-Newton methods with and without resampling.
			\item Impact of $b$ on Noise Robustness: Analyze the Newton method's performance under noise with varying $b$.
			\item Comparison with First-Order Methods:  Compare Newton method, Gauss-Newton method with Wirtinger Flow (WF) \cite{WF} and Perturbed Amplitude Flow (PAF) \cite{paf} algorithms.
		\end{itemize}
		
		\subsection{Effect of Resampling}
		In the theoretical analysis of the algorithm, it is essential that the current iteration point $\vz_k$ remains independent of the measurements $\va_j$. This independence is crucial for deriving the relevant concentration inequalities, which form the basis for establishing convergence guarantees. To evaluate the practical impact of resampling, we conduct numerical experiments comparing the performance of Newton, Gauss-Newton, Newton with resampling, and Gauss-Newton with resampling methods under two scenarios:
		\begin{itemize}
			\item The ground truth signal is randomly generated as $\vx \,{\sim}\, \mathcal{CN}(0, I_n)$.
			\item The ground truth signal is set to be the first measurement vector, i.e., $ \vx = \va_1$.
		\end{itemize}
		
		We use Gaussian model (\ref{gaussian}) with $n=512$ and $m=4n$. The total number of iterations is set to $T = 6$. For the resampling versions, the remaining $m(T-1)$ measurements are randomly generated in the same manner, i.e., $\va_j \overset{i.i.d.}{\sim}\mathcal{CN}(0, I_n/2)$. 
		\begin{figure}[htbp]
			\begin{center}
				\subfigure[]{
					\includegraphics[width=0.48\textwidth]{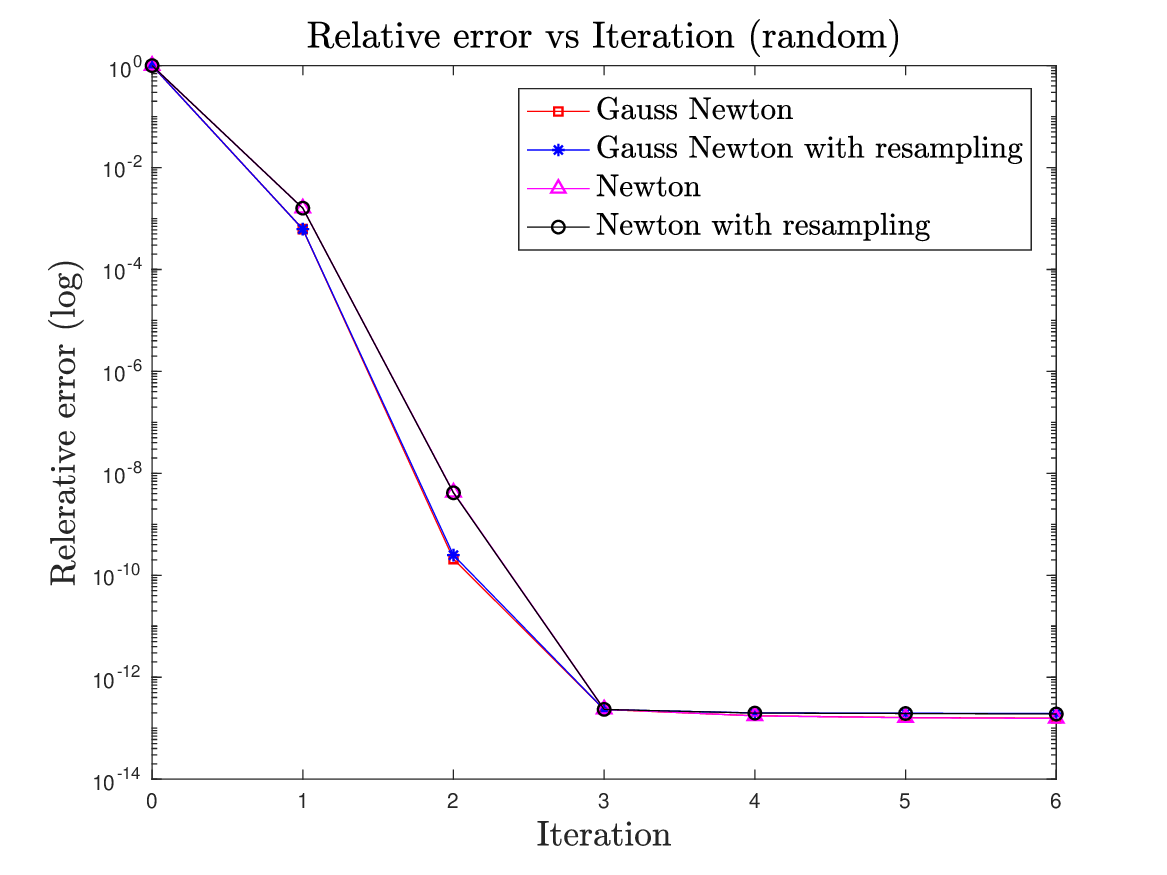}}
				\subfigure[]{
					\includegraphics[width=0.48\textwidth]{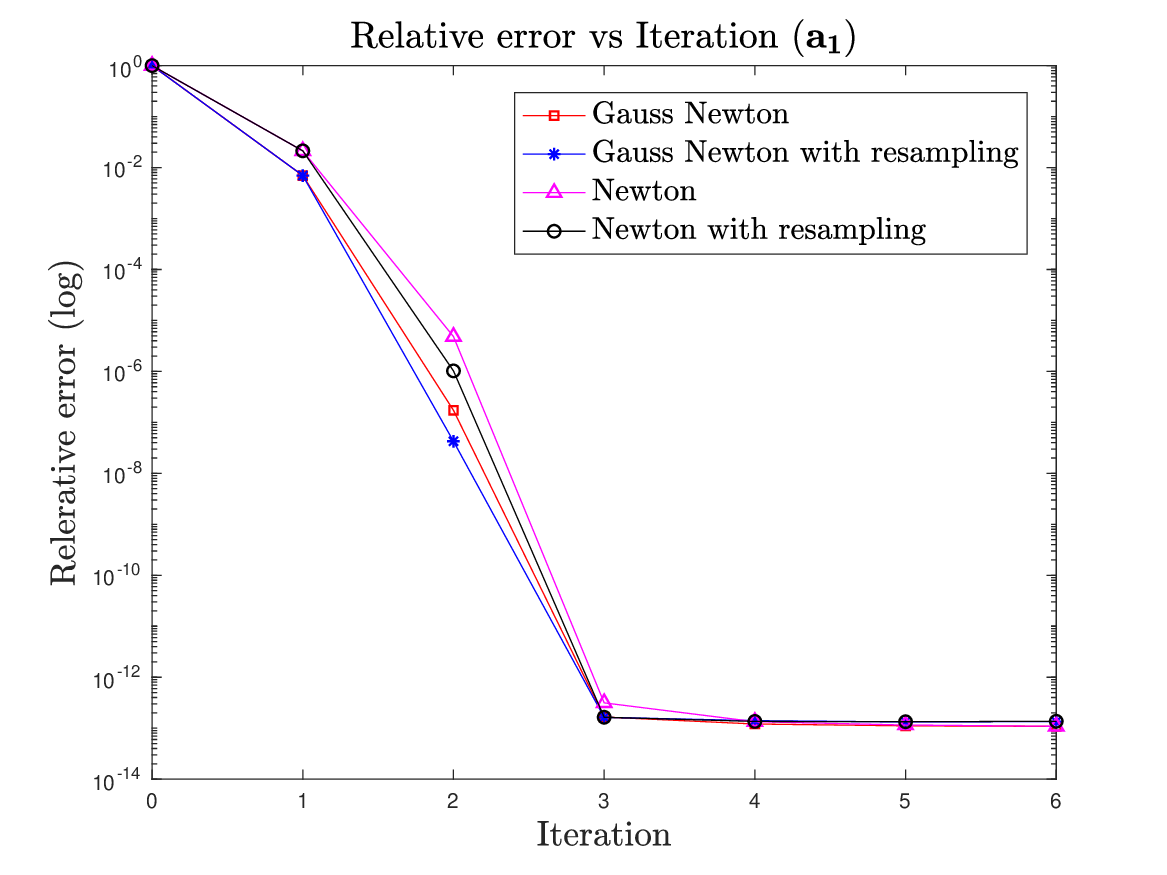}}
				\caption{Convergence of Newton and Gauss-Newton methods with/without resampling: record the relative error of each iteration step. (a) $\vx$ is randomly generated and (b) $\vx=\va_1$.}\label{1resample}
			\end{center}
		\end{figure}
		The results are presented in Figure \ref{1resample}. When both the observations and the true signal follow an independent Gaussian distribution, the results obtained with and without resampling are nearly indistinguishable. This suggests that resampling does not significantly impact performance in this scenario. While when the true signal coincides with one of the measurement vectors, the resampling versions exhibit slightly improved performance. However, the overall difference between resampling and non-resampling remains minimal, likely due to the small scale and simplicity of the problem.
		
		We hypothesize that resampling may play a more significant role in larger-scale or more complex scenarios, where the independence of measurements and iteration points becomes more critical. However, based on the current experimental results, we conclude that resampling does not offer a substantial advantage in the tested cases. Consequently, we omit the resampling versions from subsequent comparisons in the following experiments.
		
		\subsection{Impact of $b$ on Noise Robustness}\label{impactb}
		Based on the analysis of Lemma \ref{lemma_mineig}, we observe that the strong convexity of the objective function in the presence of noise is influenced by the magnitude of $b$. To investigate this relationship, we analyze the robustness of the Newton method to noise by fixing the noise level and varying $b$. We focus solely on the Newton method here, as experiments indicated similar final accuracy across all other algorithms under noise.
		\begin{figure}[!h]
			\begin{center}
				\subfigure[]{
					\includegraphics[width=0.48\textwidth]{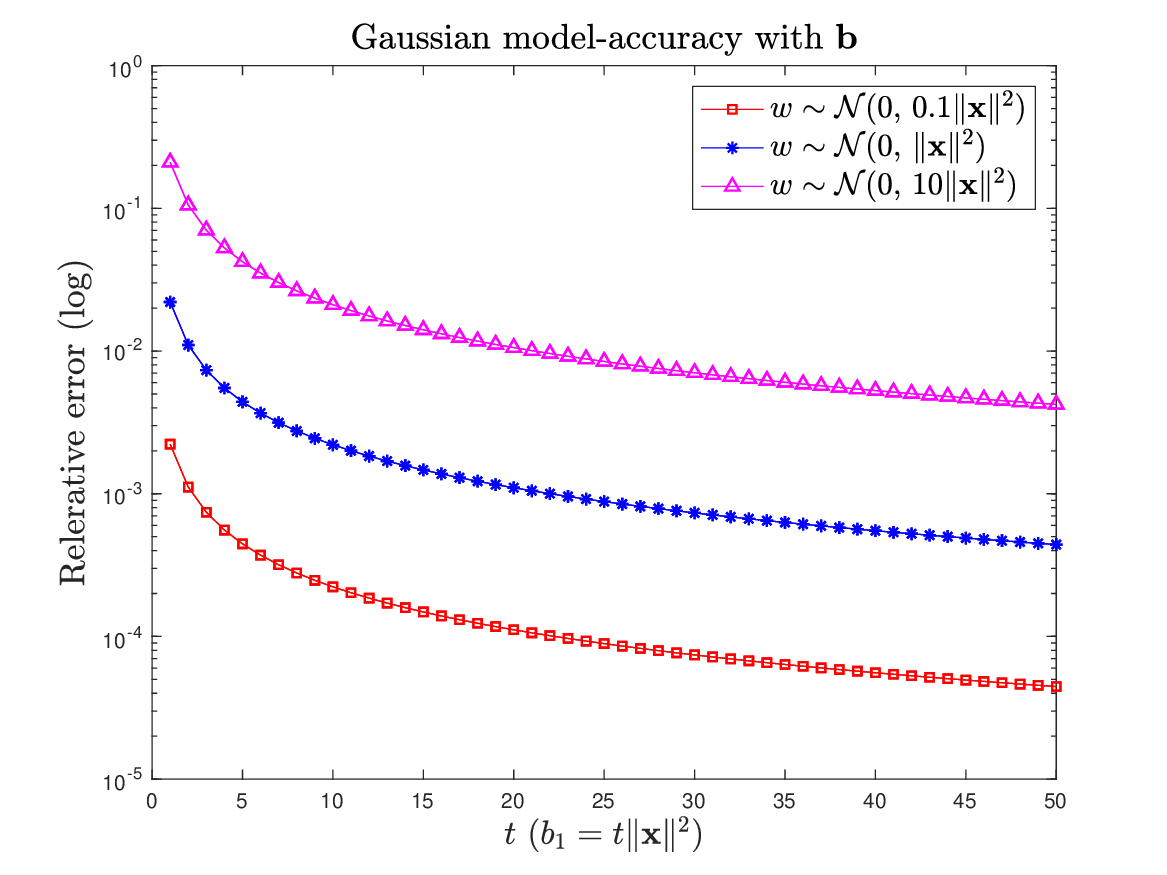}}
				\subfigure[]{
					\includegraphics[width=0.48\textwidth]{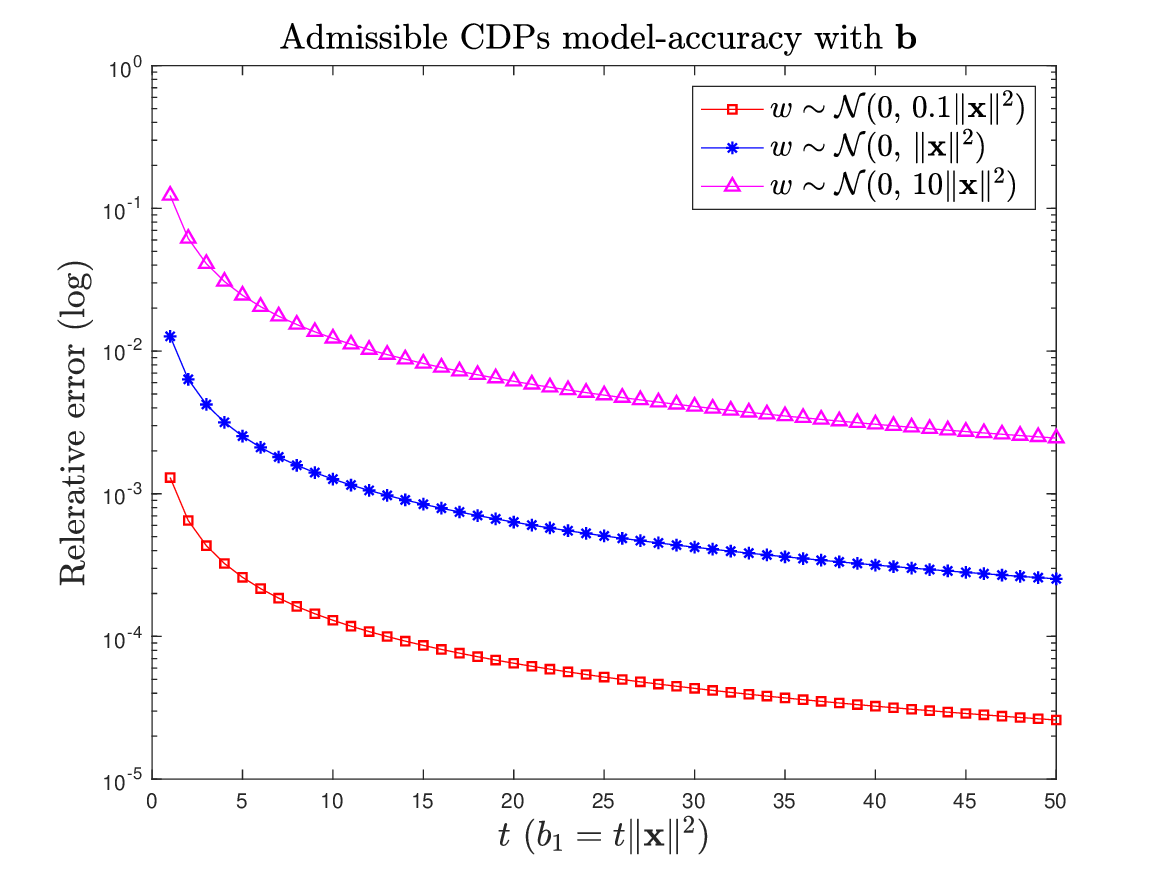}}
				\caption{Accuracy under noise: record the relative error of Newton method  for different values of $b=t\|\vx\|^2$.  (a) Gaussian model  with $ m=4n $. (b) Admissible CDPs model with $L=8$. }\label{b_noise}
			\end{center}
		\end{figure}
		In this experiment, we introduce additive noise $w_j\overset{i.i.d.}{\sim}\mathcal{N}(0,\sigma^2\|\vx\|^2)$ with $\sigma^2 \in\{0.1, 1, 10\}$. The $b$ is set as $ b = t\|\vx\|^2$, where $t$ ranges from $1 $ to $50$. For each $b$, we record the final error after 6 iterations of the Newton method and present the results in \ref{b_noise}. The results indicate that larger values of $b$ contribute to improved accuracy. For example, in Gaussian model when $\sigma^2 = 1$, the relative error decreases from $0.02$ to $4.4\times 10^{-4}$ when $b$ increases from $\|\vx\|^2$ to $50\|\vx\|^2$.
		These findings align with Lemma \ref{lemma_mineig}, which suggests that larger $b$ strengthens the strong convexity of the objective function, thereby improving noise tolerance. A detailed theoretical analysis of the relationship between $b$, noise tolerance, and convergence stability will be addressed in future work.
		
		\subsection{Comparison with First-Order Methods}
		We first adapt the Wirtinger Flow (WF) and Perturbed Amplitude Flow (PAF) algorithms for affine phase retrieval. Using the same symbol: 
		\begin{itemize}
			\item The WF method for affine phase retrieval:
			\begin{align}\label{iteration_wf}
				\vz_{k+1} &= \vz_{k} - \alpha_k \cdot \nabla f(\vz_k,\conj{\vz_k})(1:n)\\\nonumber
				&=\vz_{k} - \alpha_k \cdot \frac{1}{m}\sum_{j=1}^{m}
				\Big(\abs{\va_j^*\vz_k+b_j}^2-y_j\Big)(\va_j^*\vz_k+b_j)\va_j,
			\end{align}
			where the step size is defined as $ \alpha_k = \min\{1-\exp(-k/330), 0.2\}\cdot\frac{m}{\|\vb\|^2}$. 
			\item The PAF method for affine phase retrieval:
			\begin{align}\label{iteration_paf}
				\vz_{k+1}
				&=\vz_{k} - \mu_k \cdot \frac{1}{m}\sum_{j=1}^{m}
				\bigg(1-\frac{\sqrt{y_j+p_j^2}}{\sqrt{\abs{\va_j^*\vz_k+b_j}^2+p_j^2}}\bigg)(\va_j^*\vz_k+b_j)\va_j.
			\end{align}
			Here, the perturbed term is set as $p_j=\sqrt{y_j}$, and the step size is fixed at $\mu_k = 1.2$.
		\end{itemize}
		
		\subsubsection{Convergence Behavior}
		
		We analyze the convergence rates of Newton, Gauss-Newton, WF, and PAF under both noiseless and noisy measurements. For the Gaussian model, we set $m=4n$, while for the admissible CDPs model, we set $L=8$. In the noisy case, the noise is modeled as $w_j\overset{i.i.d.}{\sim}\mathcal{N}(0,0.1\|\vx\|^2) $, resulting in the $j$-th measurement being $y_j = |\va_j^*\vx+b_j|^2+w_j$. 
		
		\begin{figure}[!h]
			\begin{center}
				\subfigure[]{
					\includegraphics[width=0.48\textwidth]{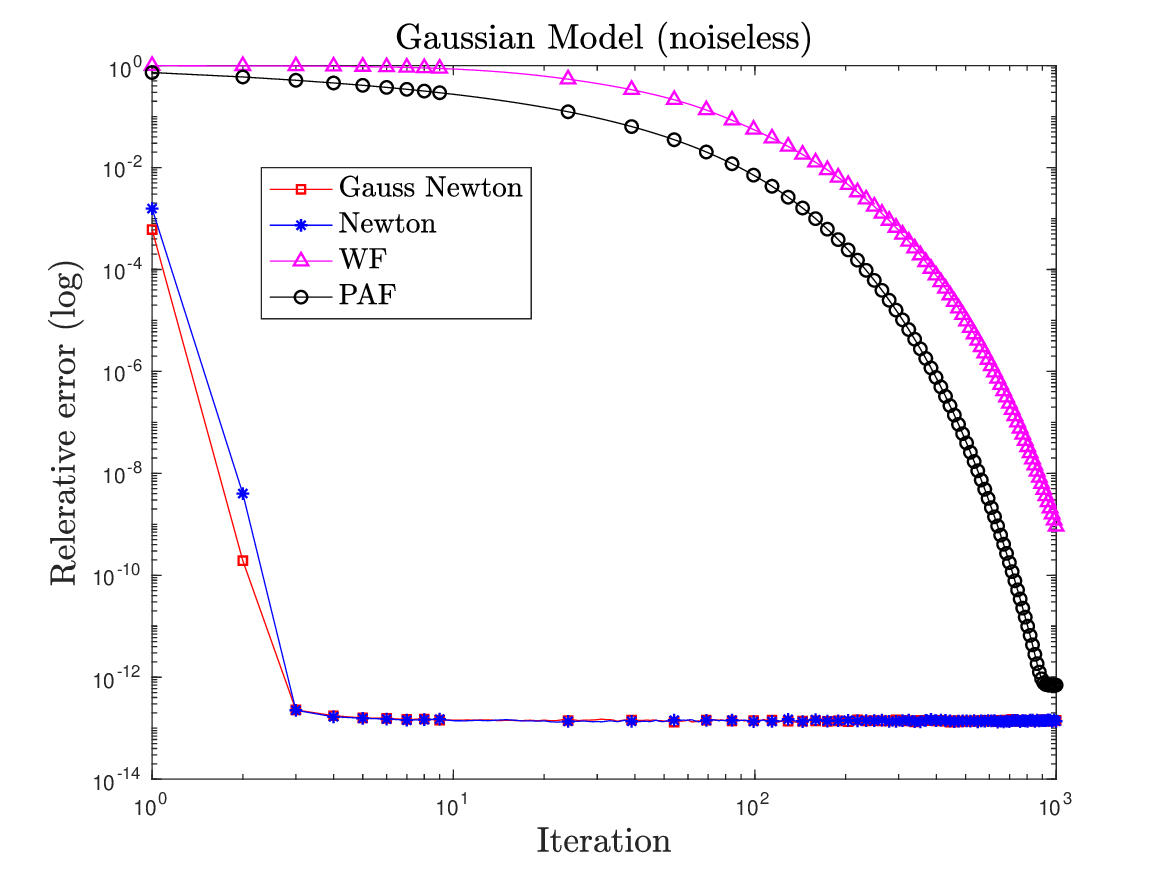}}
				\subfigure[]{
					\includegraphics[width=0.48\textwidth]{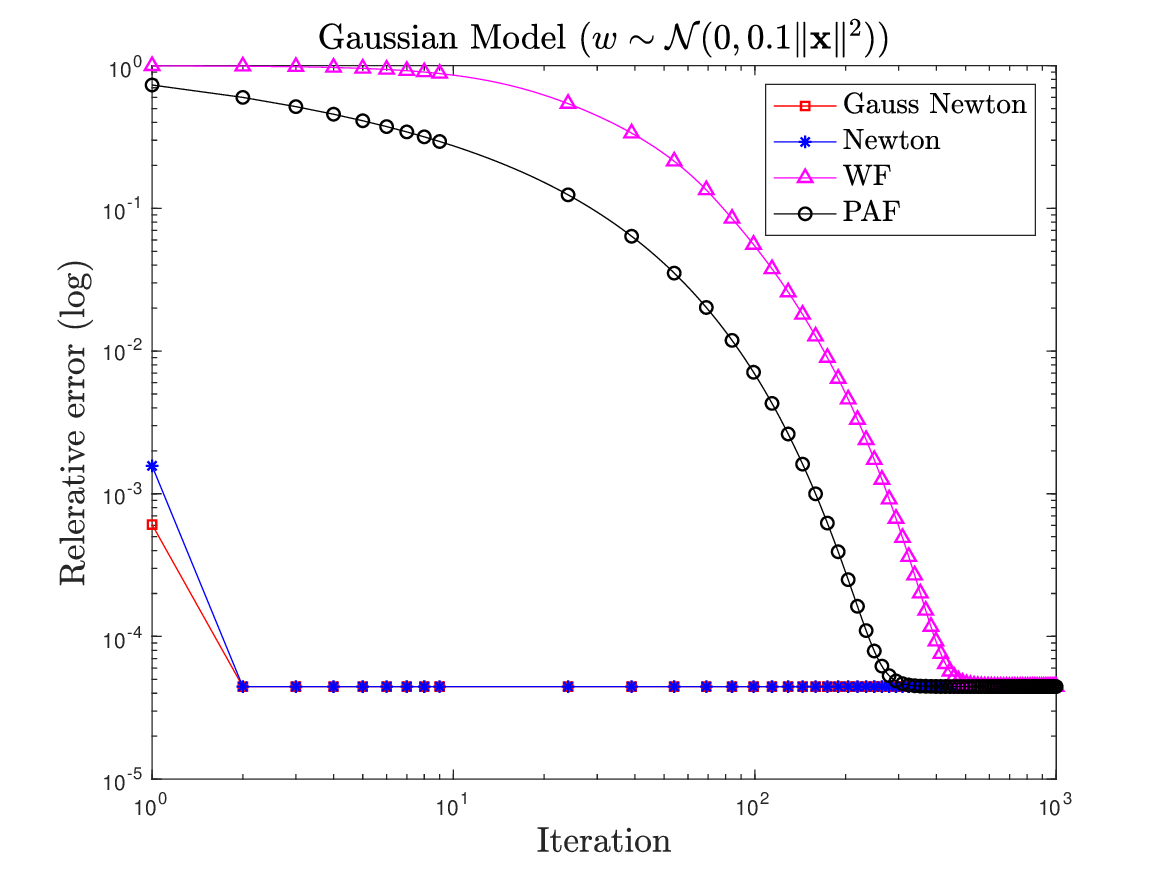}}
				\subfigure[]{
					\includegraphics[width=0.48\textwidth]{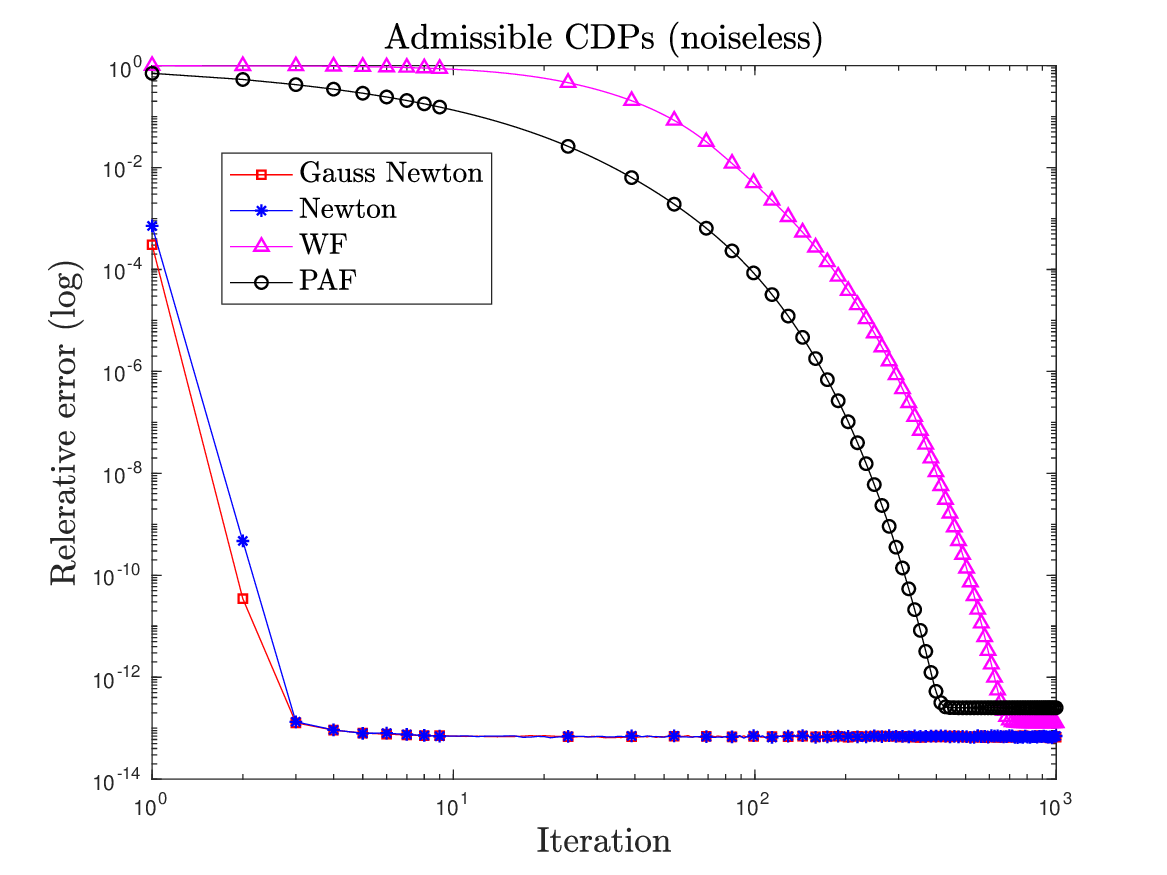}}
				\subfigure[]{
					\includegraphics[width=0.48\textwidth]{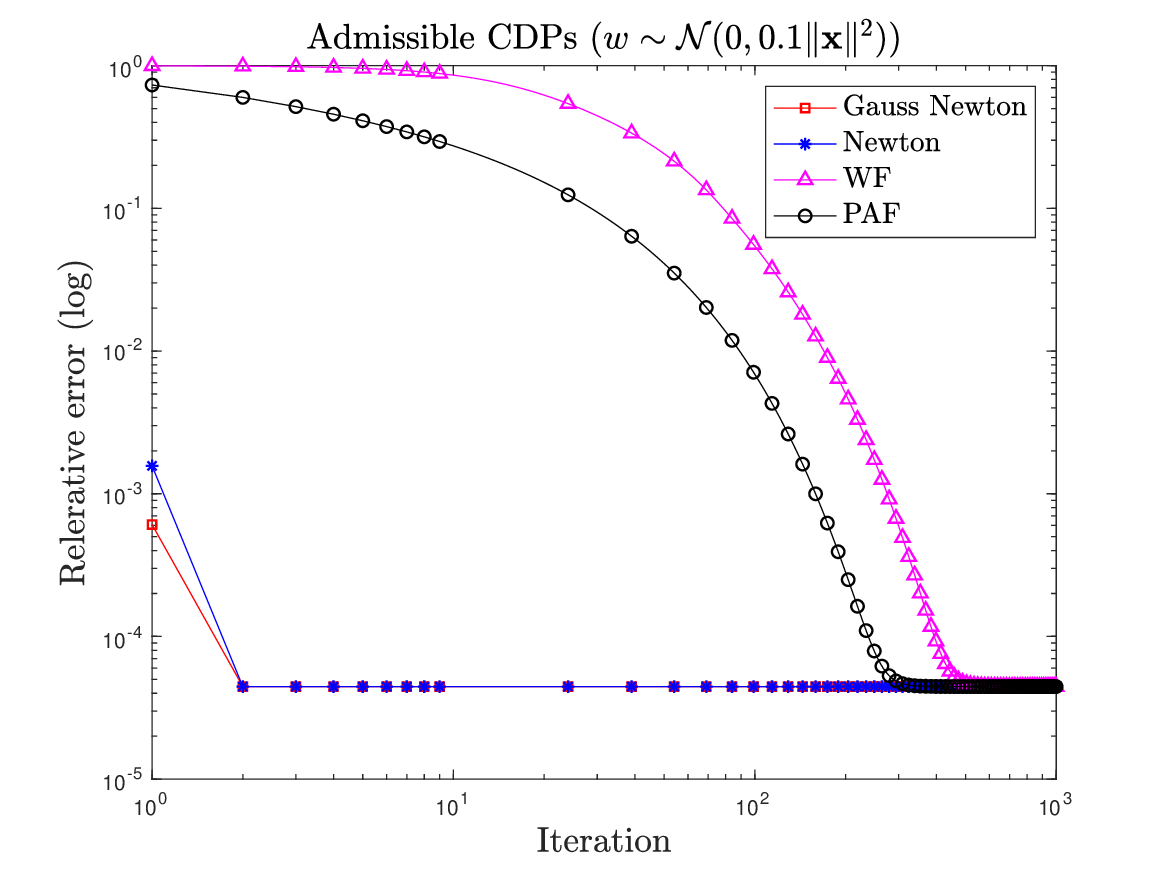}}
				\caption{Convergence experiments: relative error versus iteration count.  Here $ n = 512 $ and $ m=4n $ for the Gaussian model: (a) noiseless case, (b) Gaussian noise case. For the  admissible CDPs model with $L=8$: (c) noiseless case, (d) Gaussian noise case. }\label{converg_c}
			\end{center}
		\end{figure}
		As shown in Figure \ref{converg_c}, Newton and Gauss-Newton exhibit faster convergence due to their quadratic convergence rates. Even in the presence of noise, the second-order methods maintain superior performance, achieving similar final accuracy compared to first-order methods.
		\subsubsection{Computational Time}
		To evaluate computational efficiency, we record the time and iterations required for each method to achieve the relative error below $10^{-8}$. The results,  averaged over 50 independent trials, are presented in Table \ref{cputime_G} (Gaussian model) and Table \ref{cputime_C} (admissible CDPs model).
		\begin{table}[htb]		
			\begin{center}
				\begin{tabular}{|c|c|c|c|c|}
					\hline
					& Gauss-Newton & Newton & WF & PAF  \\\hline
					Iter & 2 & 2 & 890 & 567 \\\hline		
					Error($10^{-8}$) &  0.0213 &  0.4270   &  0.9905  & 0.9866\\\hline
					Time(s) &   0.4565  & 0.4529  &  1.3321   & 0.8459   \\\hline
				\end{tabular}
			\end{center}
			\caption{Iteration count and runtime for the Gaussian model}
			\label{cputime_G}
		\end{table}
		\begin{table}[!htb]		
			\begin{center}
				\begin{tabular}{|c|c|c|c|c|}
					\hline
					& Gauss-Newton & Newton & WF & PAF  \\\hline
					Iter & 2 & 2 & 475 & 291 \\\hline		
					Error($10^{-8}$) &  0.0045  &  0.0694    &  0.9806 &  0.9714\\\hline
					Time(s) &    0.6155  &  0.6286   &  1.3742   & 0.8667   \\\hline
				\end{tabular}
			\end{center}
			\caption{Iteration count and runtime for the admissible CDPs model}
			\label{cputime_C} 
		\end{table}
		The Gauss-Newton and Newton methods require significantly fewer iterations due to their rapid convergence. However, each iteration involves computationally expensive matrix inversions. Despite this, the overall runtime remains competitive with first-order methods, as the reduced iteration count offsets the higher per-iteration cost.
		
		\subsubsection{Success Rate}
		
		Theoretical analysis suggests that accurate recovery in the Gaussian model requires $ \mathcal{O}( n\log n )$ measurements, while the admissible CDPs model requires $L= \mathcal{O}(\log^3n) $.
		To validate these theoretical bounds, we conducted experiments by varying the number of measurements and recording the success rate for each $ m/n $ (Gaussian model) or $L$ (admissible CDPs model). A trial was considered successful if the relative error fell below $ 10^{-5} $ within a maximum of 10 iterations for Newton and Gauss-Newton and 1500 iterations for WF and PAF. For each $ m/n $ or $ L $, we performed 50 independent trials to calculate the success rate.
		\begin{figure}[htb]
			\begin{center}
				\subfigure[]{
					\includegraphics[width=0.48\textwidth]{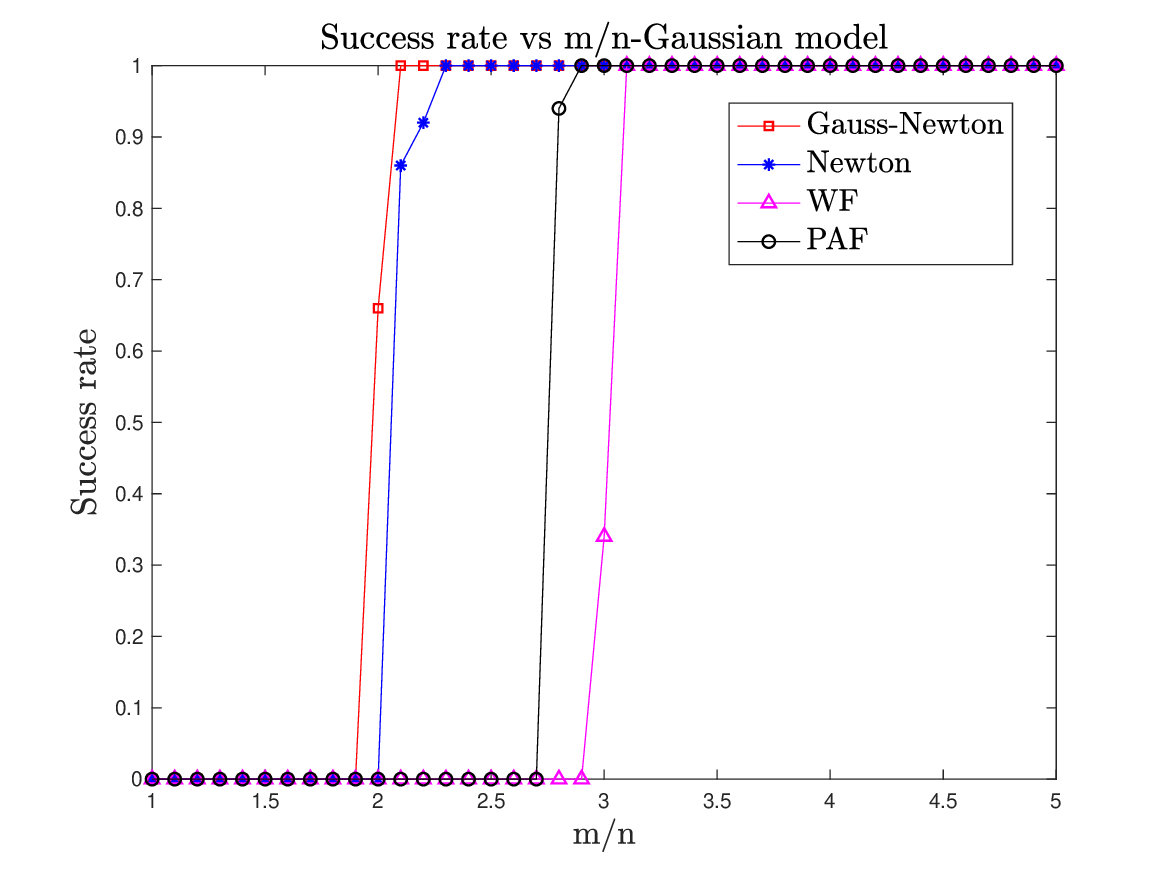}}
				\subfigure[]{
					\includegraphics[width=0.48\textwidth]{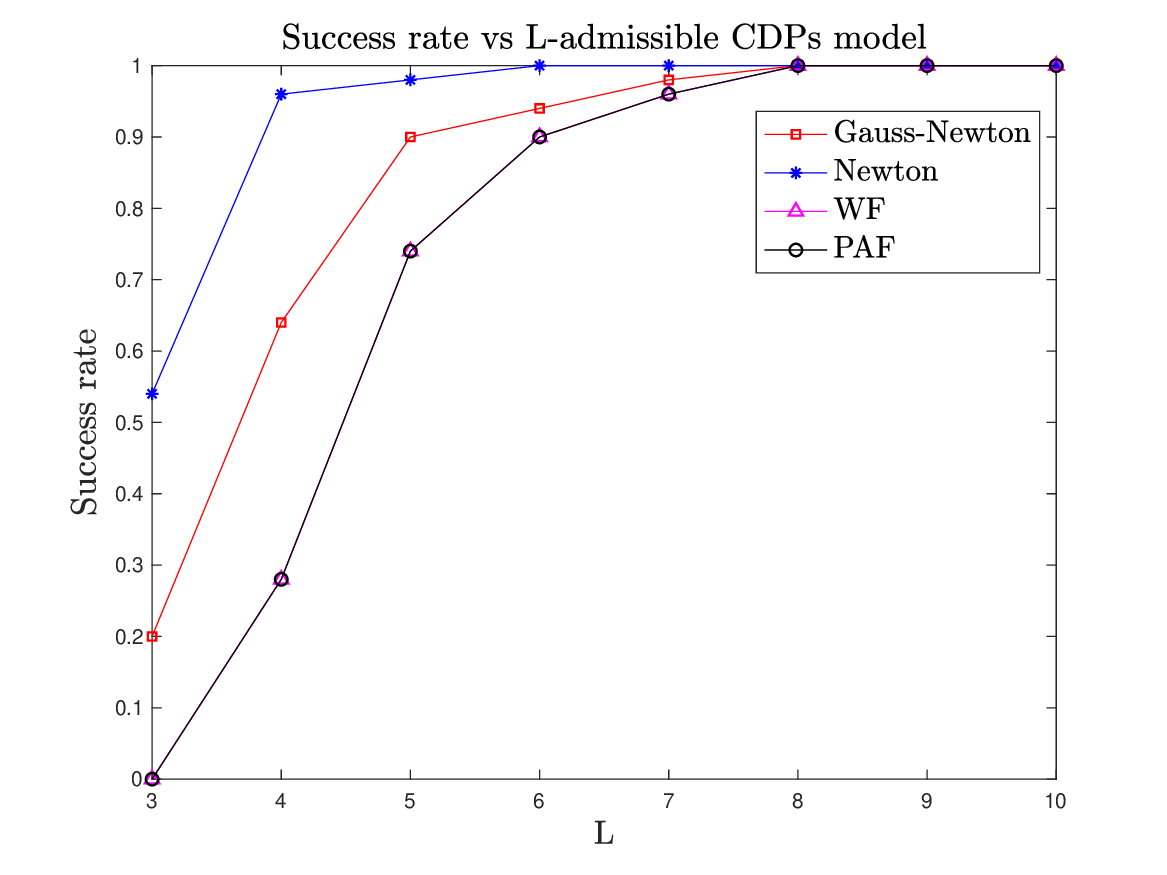}}
				\caption{Success rate versus measurement count. (a) Gaussian model with $ m/n=1:0.1:5 $.  (b) Admissible CDPs model with $ L\in[3,10]$.}\label{success}
			\end{center}
		\end{figure}
		
		In the Gaussian model, $ m/n $ was varied from $1 $ to $ 5$ with a step size of $ 0.1 $, while in the admissible CDPs model, $ L $ was varied from $3$ to $10$ with a step size of $ 1 $. The results, displayed in Figure \ref{success}, demonstrate that the practical number of measurements required for successful recovery is significantly lower than the theoretical bounds. This indicates that the algorithm is robust and efficient, even with a reduced number of observations.
		
		
		\section{Conclusion}
		In this paper, we reformulate the affine phase retrieval problem as  a nonlinear least-squares optimization problem and develop second-order algorithms based on Newton and Gauss-Newton methods. 
		Unlike traditional phase retrieval that suffers from non-convexity and sensitivity to initialization, we prove that the affine formulation exhibits strong convexity under specific signal-dependent conditions (e.g., when the affine shift $|b|^2$ exceeds a problem-dependent threshold). This property enables the direct application of second-order methods without requiring careful initialization, a significant advantage over conventional phase retrieval.
		Specifically, we prove that with noiseless Gaussian or admissible CDPs measurements, the Newton method with resampling achieves quadratic convergence to the exact solution. 
		
		The analytical techniques developed in this study are not limited to the Newton and Gauss-Newton method but can also be extended to derive convergence guarantees for other second-order algorithms in the context of affine phase retrieval.
		\section{Acknowledgments}
		The author would like to thank Prof. Zhiqiang Xu for his helpful discussions. This work is supported by NSFC grant \#12001297 .
		\section{Appendix}
		\begin{lemma}\label{affine_lemma}[Theorem 3.1 in \cite{affine_g}]
			Let $ A = (\va_1,\ldots,\va_m)^\top\in\C^{m\times n} $ and $ \vb=(b_1,\ldots,b_m)^\top\in\C^m $. Let $ \MM_{A,\vb}^2: \C^n \to \R^m $ be a map defined by 
			\[
			\MM_{A,\vb}^2:=\big(|\va_1^*\vx+b_1|^2,\ldots,|\va_m^*\vx+b_m|^2\big) 
			\]
			Then, $ (A,\vb) $ is affine phase retrievable for $ \C^n $ if and only if the real Jacobian of map $ \MM_{A,\vb}^2 $ has rank $ 2n $ for all $ \vx\in\C^n $.
		\end{lemma}
		\begin{lemma}\label{WF_lemma}[Lemma 7.4 in \cite{WF}]
			Let $\vx \in \C^n$ be a signal with $\|\vx\| = 1$, which is independent of the measurement samples $\va_j \in \C^n$, for $j = 1, 2, \ldots, m$. Assume the vectors $\va_j$ are drawn according to either the Gaussian model or the admissible CDPs model. Define the matrix
			\[
			W := \frac{1}{m} \sum_{j=1}^{m}
			\begin{pmatrix}
				|\va_j^* \vx|^2 \va_j \va_j^* & (\va_j^* \vx)^2 \va_j \va_j^\top \vspace{3mm} \\
				(\vx^* \va_j)^2 \conj{\va_j} \va_j^* & |\va_j^* \vx|^2 \conj{\va_j} \va_j^\top
			\end{pmatrix}.
			\]
			For any $\epsilon > 0$, if the number of samples satisfies $m \geq C_\epsilon n \log n$ in the Gaussian model, and the number of patterns satisfies $L \geq C_\epsilon \log^3 n$ in the admissible CDPs model, with a sufficiently large constant $C_\epsilon$, 
			\[
			\|W - \E[W]\| \leq \epsilon \|\E[W]\|,
			\]
			holds with probability at least $1 - 10\exp(-\gamma_\epsilon n) - 8/n^2$ for the Gaussian model, and $1 - (2L + 1)/n^3$ for the admissible CDPs model, where $\gamma_\epsilon > 0$ is a constant.
		\end{lemma}
		\begin{lemma}\label{WF_lemma1}[Lemma 7.8 in \cite{WF}]
			Under the same assumptions as  Lemma \ref{WF_lemma}, 
			\[
			\Big\|I_n - \frac{1}{m}\sum_{j=1}^{m}\va_j\va_j^*\Big\|\leq \epsilon,\quad \Big\|\frac{1}{m}\sum_{j=1}^{m}\va_j\va_j\zz\Big\|\leq \epsilon
			\]
			holds with probability at least $ 1-2\exp(-\gamma_\epsilon m) $ for the Gaussian model, and at least $ 1-1/n^2 $ for the admissible CDPs model.  Here, $\gamma_\epsilon > 0$ is a constant that depends on $\epsilon$.
		\end{lemma}
		\begin{lemma}\label{expectations}
			
			Let $\vx \in \C^n$ be a signal with $\|\vx\| = 1$, which is independent of the measurement samples $\va_j$, for $j = 1, \ldots, m$. Assume that the vectors $\va_j$ are drawn according to either the Gaussian model or the admissible CDPs model. Let $\vb = (b_1, b_2, \ldots, b_m) \in \C^m$ be a fixed vector, with each entry equal to $b \in \C$. Define the matrix
			\[
			S := \frac{1}{m} \sum_{j=1}^{m} 
			\begin{pmatrix}
				|\va_j^*\vx + b_j|^2 \va_j \va_j^* & (\va_j^* \vx + b_j)^2 \va_j \va_j^\top \vspace{3mm} \\
				(\vx^* \va_j + \conj{b_j})^2 \conj{\va_j} \va_j^* & |\va_j^*\vx + \conj{b_j}|^2 \conj{\va_j} \va_j^\top
			\end{pmatrix}.
			\]
			Taking the expectation of $S$, we get:
			\[
			\E[S] = 
			\begin{pmatrix}
				(\|\vx\|^2 + |b|^2) I_n + \vx \vx^* & 2 \vx \vx^\top \vspace{2mm} \\
				2 \conj{\vx} \vx^* & (\|\vx\|^2 + |b|^2) I_n + \conj{\vx} \vx^\top
			\end{pmatrix}.
			\]
		\end{lemma}
		\begin{proof}
			{\bfseries 1. The Gaussian model:}
			
			Since the vector $ \vx $ and scalars $ b_j=b,\,j=1,\ldots,m $ are all independent of $ \va_j,\,j=1,\ldots,m $, the random variable $ \abs{\va_1^*\vx+b_1}^2\va_1\va_1^* $ has the same expectation as $ \abs{\va^*\vx+b}^2\va\va^* $, where $ \va=(a_1,\ldots,a_n)^\top\in\C^n $ is any Gaussian random vector. 
			By definition, we have
			\[
			\abs{\va^*\vx+b}^2=\big(\sum_{i=1}^{n}\conj{a_i}x_i+b\big)\big(\sum_{j=1}^{n}a_j\conj{x_j}+\conj{b}\big)=\sum_i\sum_{j}\conj{a_i}a_jx_i\conj{x_j}+\conj{b}\sum_{i}\conj{a_i}x_i+b\sum_{j}a_j\conj{x_j}+\abs{b}^2.
			\]
			Expanding this into the matrix $ \abs{\va^*\vx+b}^2\va\va^* $, we get:
			\begin{align*}
				\abs{\va^*\vx+b}^2\va\va^*=\Big(\sum_i\sum_{j}\conj{a_i}a_jx_i\conj{x_j}+\conj{b}\sum_{i}\conj{a_i}x_i+b\sum_{j}a_j\conj{x_j}+\abs{b}^2\Big)\begin{pmatrix}
					\abs{a_1}^2 & a_1\conj{a_2}&\cdots& a_1\conj{a_n}\\
					a_2\conj{a_1} & \abs{a_2}^2 & \cdots& a_2\conj{a_n}\\
					\vdots & \vdots & \ddots & \vdots\\
					a_n\conj{a_1} & a_n\conj{a_2} &\cdots& \abs{a_n}^2
				\end{pmatrix}.
			\end{align*}
			Denote this matrix by $P=(p_{i,j})_{i,j\in[1:n]}$. For the diagonal elements, we compute the expectation:
			\begin{align*}
				\E[p_{k,k}]&=\E\Big[\big(\sum_i\sum_{j}\conj{a_i}a_jx_i\conj{x_j}+\conj{b}\sum_{i}\conj{a_i}x_i+b\sum_{j}a_j\conj{x_j}+\abs{b}^2\big)|a_k|^2\Big]\\
				&=\E\Big[\sum_{i=j=k}|a_k|^4|x_k|^2+\sum_{i=j\neq k}|a_i|^2|a_k|^2|x_i|^2+|b|^2|a_k|^2\Big]\\
				&=2|x_k|^2+\sum_{i\neq k}|x_i|^2+|b|^2\\
				&=\|\vx\|^2+|b|^2+|x_k|^2. 
			\end{align*}
			Recall that each $ a_i\in\C\overset{i.i.d.}{\sim} \mathcal{CN}(0,1/2) $, and thus for $k\neq g$.
			\[
			\E[a_k \conj{a_g}]=\E[|a_k^2|\cdot\conj{a_g}]=0.
			\]
			Therefore, for the off-diagonal elements $  p_{k,g} (k\neq g )$, we have:
			\begin{align*}
				\E[p_{k,g}]&=\E\Big[\big(\sum_i\sum_{j}\conj{a_i}a_jx_i\conj{x_j}+\conj{b}\sum_{i}\conj{a_i}x_i+b\sum_{j}a_j\conj{x_j}+\abs{b}^2\big)a_k\conj{a_g}\Big]\\
				&=\E\Big[\sum_{i=k,j=g}|a_k|^2|a_g|^2 x_k\conj{x_g}\Big]\\
				&=x_k\conj{x_g}.
			\end{align*}
			Combining the above two results, we obtain
			\[
			\E[\abs{\va_j^*\vx+b_j}^2\va_j\va_j^*]=\E[\abs{\va^*\vx+b}^2\va\va^*]=(\|\vx\|^2+\abs{b}^2)I_n+\vx\vx^*.
			\]
			Similarly, we calculate:
			\[
			\E\big[(\va_j^*\vx+b_j)^2\va_j\va_j^\top\big]=2\vx\vx^\top,
			\]  
			which leads to the final conclusion:
			\begin{align*}
				\E[S]&=\begin{pmatrix}
					(\|\vx\|^2+|b|^2)I_n+\vx\vx^* & 2\vx\vx^\top \vspace{2mm}\\
					2\conj{\vx}\vx^* & (\|\vx\|^2+|b|^2)I_n+\conj{\vx}\vx^\top 
				\end{pmatrix}.
			\end{align*}
			{\bfseries 2. The admissible CDPs model:}
			In the admissible CDPs model (\ref{cdp}), the measurement vectors are structured as $\va_j = D_l \vf_k$, where $D_l$ is a diagonal matrix with i.i.d. entries drawn from a distribution $g$ that $ \E[|g|^2]=1 $ and $|g|<\sqrt{6}$, and $\vf_k$ is the $k$-th column of the discrete Fourier transform matrix. 
			Correspondingly, let $ b_j =b_{(l,k)}=b $. Here $ j=1,2,\ldots,m $, $m=nL$, $ l= \lceil j/n\rceil,$, and $ k=(j\,\text{mod}\,n)-1$. Then the matrix $ S $ can be formulated as
			\[
			S=\frac{1}{L}\sum_{l=1}^{L}\begin{pmatrix}
				\frac{1}{n}\sum_{k=1}^{n}\abs{\vf_k^*D_l^*\vx+b_{(k,l)}}^2D_l\vf_k\vf_k^* D_l^*& \frac{1}{n}\sum_{k=1}^{n}(\vf_k^*D_l^*\vx+b_{(k,l)})^2D_l\vf_k\vf_k^\top D_l^\top \vspace{3mm}\\
				\frac{1}{n}\sum_{k=1}^{n}(\conj{\vf_k^*D_l^*\vx}+\conj{b}_{(k,l)})^2\conj{D}_l\conj{\vf}_k\vf_k^* D_l^*& \frac{1}{n}\sum_{k=1}^{n}\abs{\vf_k^*D_l^*\vx+b_{(k,l)}}^2\conj{D}_l\conj{\vf}_k\vf_k^\top D_l^\top
			\end{pmatrix}.
			\]
			To compute the expectation of $ \frac{1}{n}\sum_{k=1}^{n}\abs{\vf_k^*D_l^*\vx+b_{(k,l)}}^2D_l\vf_k\vf_k^* D_l^* $, it is sufficient to calculate the expectation of $ \frac{1}{n}\sum_{k=1}^{n}\abs{\vf_k^*D^*\vx+b}^2D\vf_k\vf_k^* D^* $. Let $ \omega=\exp(2\pi i/n) $ be the $ n $-th root of unity.
			By definition, we know
			\[
			\vf_k^*=(\omega^{-0(k-1)},\omega^{-1(k-1)},\ldots,\omega^{-(n-1)(k-1)})  
			\]
			and
			\begin{align*}
				&\frac{1}{n}\sum_{k=1}^{n}\abs{\vf_k^*D^*\vx+b}^2D\vf_k\vf_k^* D^* \\
				&=\frac{1}{n}\sum_{k=1}^{n}\Big(\sum_{i,\,j}\conj{d_i}d_j\omega^{(-i+j)(k-1)}x_i\conj{x_j}+\conj{b}\sum_{i}\conj{d_i}\omega^{-(i-1)(k-1)}x_i+b\sum_{j}d_j\omega^{(j-1)(k-1)}\conj{x_j}+\abs{b}^2\Big)\\
				&\quad \quad \begin{pmatrix}
					\abs{d_1}^2 & d_1\conj{d_2}\omega^{-(k-1)}&\cdots& d_1\conj{d_n}\omega^{-(n-1)(k-1)}\\
					d_2\conj{d_1}\omega^{(k-1)} & \abs{d_2}^2 & \cdots& d_2\conj{d_n}\omega^{-(n-2)(k-1)}\\
					\vdots & \vdots & \ddots & \vdots\\
					d_n\conj{d_1}\omega^{(n-1)(k-1)} & d_n\conj{d_2}\omega^{(n-2)(k-1)} &\cdots& \abs{d_n}^2
				\end{pmatrix}.
			\end{align*}
			Therefore, through calculation, the expected value of the $ s $-th diagonal element is
			\begin{align*}
				&\frac{1}{n}\sum_{k=1}^{n}\E\Big[\big(\sum_{i,\,j}\conj{d_i}d_j\omega^{(-i+j)(k-1)}x_i\conj{x_j}+\conj{b}\sum_{i}\conj{d_i}\omega^{-(i-1)(k-1)}x_i+b\sum_{j}d_j\omega^{(j-1)(k-1)}\conj{x_j}+\abs{b}^2\big)|d_s|^2\Big]\\
				&=\E\Big[\sum_{i=j=s}|d_s|^4|x_s|^2+\sum_{i=j\neq s}|d_i|^2|d_s|^2|x_i|^2+|b|^2|d_s|^2\Big]\\
				&=2|x_s|^2+\sum_{i\neq s}|x_i|^2+|b|^2\\
				&=\|\vx\|^2+|x_s|^2+|b|^2.
			\end{align*}  
			For the off-diagonal elements ($ s\neq t $), we compute the expectation as
			\begin{align*}
				&\frac{1}{n}\sum_{k=1}^{n}\E\Big[\big(\sum_{i,\,j}\conj{d_i}d_j\omega^{(-i+j)(k-1)}x_i\conj{x_j}+\conj{b}\sum_{i}\conj{d_i}\omega^{-(i-1)(k-1)}x_i+b\sum_{j}d_j\omega^{(j-1)(k-1)}\conj{x_j}+\abs{b}^2\big)d_g\conj{d_h}\Big]\\
				&=\frac{1}{n}\sum_{k=1}^{n}\E\Big[\sum_{i=s,j=t}|d_s|^2|d_t|^2\omega^{(-s+t)(k-1)}x_s\conj{x_t}\Big]\\
				&=x_s\conj{x_t}.
			\end{align*}
			The results imply
			\[
			\E\Big[\frac{1}{n}\sum_{k=1}^{n}\abs{\vf_k^*D^*\vx+b}^2D\vf_k\vf_k^* D^*\Big]=(\|\vx\|^2+\abs{b}^2)I_n+\vx\vx^*.
			\]
			Similarly, the expectation of $ \frac{1}{n}\sum_{k=1}^{n}(\vf_k^*D_l^*\vx+b_{(k,l)})^2D_l\vf_k\vf_k^\top D_l^\top  $ can be computed in a similar fashion, yields the final conclusion.
		\end{proof}
		\begin{lemma}\label{expectation_lemma}
			Under the setup of Lemma \ref{expectations}, define the matrix
			\[
			S := \frac{1}{m}\sum_{j=1}^{m}\begin{pmatrix}
				\abs{\va_j^*\vx+b_j}^2\va_j\va_j^* & (\va_j^*\vx+b_j)^2\va_j\va_j^\top \vspace{3mm}\\
				(\vx^*\va_j+\conj{b_j})^2\conj{\va_j}\va_j^* & \abs{\va_j^*\vx+b_j}^2\conj{\va_j}\va_j^\top
			\end{pmatrix}.
			\]
			For any $\epsilon > 0$ and some constant $\gamma_\epsilon > 0$, if the number of samples in the Gaussian model satisfies $m \geq C_\epsilon n \log n$, and the number of patterns in the admissible CDPs model satisfies $L \geq C_\epsilon \log^3 n$, where $C_\epsilon$ is a sufficiently large constant, then
			\[
			\|S - \E[S]\| \leq \epsilon \|\E[S]\|
			\]
			holds with probability at least $1 - 20 \exp(-\gamma_\epsilon n) - 16/n^2$ for the Gaussian model, and $1 - 2(2L + 1)/n^3 - 2/n^2$ for the admissible CDPs model.
		\end{lemma}
		\begin{proof}	
			By comparing this matrix with the matrix $W$ given in Lemma \ref{WF_lemma}, we rewrite
			\[
			S = W + S_1 + S_2,
			\]
			where
			\[
			S_1:= \frac{1}{m}\sum_{j=1}^{m}\begin{pmatrix}
				\abs{b_j}^2\va_j\va_j^*& (b_j)^2\va_j\va_j^\top \vspace{3mm}\\
				(\conj{b_j})^2\conj{\va_j}\va_j^*& \abs{b_j}^2\conj{\va_j}\va_j^\top
			\end{pmatrix}
			\]
			and
			\[ S_2 := \frac{2}{m}\sum_{j=1}^{m}\begin{pmatrix}
				\Re(\conj{b_j}\cdot \va_j^*\vx)\va_j\va_j^*& (b_j\cdot \va_j^*\vx)\va_j\va_j^\top \vspace{3mm}\\
				(\conj{b_j}\cdot \conj{\va_j^*\vx})\conj{\va_j}\va_j^*& \Re(\conj{b_j}\cdot \va_j^*\vx)\conj{\va_j}\va_j^\top
			\end{pmatrix}.
			\] 
			From Lemma \ref{WF_lemma} and Lemma \ref{WF_lemma1}, we have
			\[
			\|W-\E[W]\|\leq \epsilon/3 \|\E[W]\|
			\]
			and 
			\[
			\|S_1 - \E[S_1]\|\leq \epsilon/3 \|\E[S_1]\|
			\]
			with probability at least $ 1-12\exp(-\gamma_\epsilon n)-8/n^2 $ for the Gaussian model, and $ 1-(2L+1)/n^3 -1/n^2$ for the admissible CDPs model.
			Thus, to complete the proof, it remains to show that
			\[
			\|S_2 - \E[S_2]\|\leq \epsilon/3.
			\]
			To establish this, we need to prove the following two inequalities for a fixed $ \vx $:
			\begin{equation}\label{inequality1}
				\Big\|\frac{1}{m}\sum_{j=1}^{m}\Re(\conj{b_j}\cdot \va_j^*\vx)\va_j\va_j^*\Big\|\leq \epsilon/24,
			\end{equation}
			and 
			\begin{equation}\label{inequality2}
				\Big\|\frac{1}{m}\sum_{j=1}^{m}(b_j\cdot \va_j^*\vx)\va_j\va_j^\top\Big\|\leq \epsilon/24.
			\end{equation}
			By unitary invariance, we can assume that $ \vx=\ve_1 $. To prove inequality (\ref{inequality1}), it suffices to show that 
			\begin{align}\label{inequality12}
				T(\vv):=\Big|\frac{1}{m}\sum_{j=1}^{m}\big(\Re(b\cdot\conj{\va_j(1)})\big)|\va_j^*\vv|^2\Big|\leq \epsilon/24.
			\end{align}
			holds with high probability for any $ \vv\in \C^n $ with $ \|\vv\|=1 $. 
			\newline
			\noindent
			{\bfseries The Gaussian Model:}
			
			Following the method used in paper \cite{WF}, we first define the event $ E_0 $, where the following inequalities hold for any $ \epsilon_0>0 $:
			\[
			\frac{1}{m}\sum_{j=1}^{m}(|\va_j(1)|^2-1)\leq \epsilon_0,\quad \frac{1}{m}\sum_{j=1}^{m}(|\va_j(1)|^4-2)\leq \epsilon_0,\quad \max_{1\leq j\leq m}|\va_j(1)|\leq \sqrt{10\log m}.
			\]
			By Chebyshev's inequality, the event $ E_0 $ holds with probability at least $ 1-4n^{-2} $, provided that $ m\geq Cn $ for a sufficiently large constant $ C $.
			
			We now proceed to prove (\ref{inequality12}). For a fixed $ \vv $, let $ \vv = (\vv(1),\tilde{\vv}) $, and partition $ \va_j$ in the same way. Then 
			\begin{align*}
				T(\vv)&=\Big|\frac{1}{m}\sum_{j=1}^{m}\Re(b\cdot\conj{\va_j(1)})\cdot\big(|\conj{\va_j(1)}|^2|\vv(1)|^2+|\tilde{\va}_j^*\tilde{\vv}|^2+2\Re (\conj{\va_j(1)}\vv(1)\tilde{\va}_j^*\tilde{\vv})\big)\Big|\\
				&\leq\Big|\frac{1}{m}\sum_{j=1}^{m}\Re(b\cdot\conj{\va_j(1)})\big(|\conj{\va_j(1)}|^2|\vv(1)|^2+2\Re (\conj{\va_j(1)}\vv(1)\tilde{\va}_j^*\tilde{\vv})\big)\Big|\\ 
				&\quad\quad+\Big|\frac{1}{m}\sum_{j=1}^{m}\Re(b\cdot\conj{\va_j(1)})|\tilde{\va}_j^*\tilde{\vv}|^2\Big|\\
				&:=T_1 + T_2.
			\end{align*}
			For the term $ T_1 $, using Hoeffding's inequality, for any $ \epsilon>0 $ and $ \gamma >0 $, there exists a sufficiently large constant $ C_{(\epsilon, \gamma)} $ such that when $ m\geq C_{(\epsilon, \gamma)}|b|\sqrt{n\sum_{j=1}^{m}|\va_j(1)|^4}$,
			\[
			T_1\leq \epsilon/96
			\]  
			holds with probability at least $ 1-2\exp(-2\gamma n) $. Next, we control the term $ T_2 $ by Bernstein-type inequality. For any $ \epsilon>0 $ and $ \gamma>0 $, when $ m\geq C_{(\epsilon, \gamma)}|b|\sqrt{n\sum_{j=1}^{m}|\va_j(1)|^2}+n\max_{j=1}^{m}|\va_j(1)| $, 
			\[
			T_2\leq \epsilon/96
			\]
			holds with probability at least $ 1-2\exp(-2\gamma n) $. Thus, for a fixed $ \vv $, we have $ T(\vv)\leq \epsilon/48 $ with probability at least $ 1-4\exp(-2\gamma n) $.  
			
			Now, we extend this result to any $ \vv\in\C^n $. Let $ \mathcal{N} $ be a $ 1/4 $-net of the unit sphere in $ \C^n $.  Then
			\[
			\max_{\vv\in\C^n}T(\vv)\leq 2\max_{\vv\in\mathcal{N}}T(\vv)\leq \epsilon/24
			\]
			holds with probability at least $ 1-4\exp(-\gamma n) $, provided that $ m\geq C\big(\sqrt{n\sum_{j=1}^{m}|\va_j(1)|^4} + \sqrt{n\sum_{j=1}^{m}|\va_j(1)|^2}+n\max_{j=1}^{m}|\va_j(1)|\big) $.
			Under the event $ E_0 $, when $ m\geq Cn\log n $, the inequality (\ref{inequality12}) holds with high probability, provided $ C $ sufficiently large. Thus, (\ref{inequality1}) holds with probability at least $ 1-4\exp(-\gamma n)-4n^{-2} $.
			
			Using a similar approach, we can prove inequality (\ref{inequality2}). Therefore, $\|S_2 - \E[S_2]\|\leq \epsilon/3 $ holds with probability at least $ 1-8\exp(-\gamma n)-4n^{-2} $.

			Combining this with the earlier results $\|W-\E[W]\|\leq \epsilon/3 \|\E[W]\|$ and $\|S_1-\E[S_1]\|\leq \epsilon/3 \|\E[S_1]\|$,  we conclude that for any $ \epsilon>0 $,
			\[
			\|S-\E[S]\|\leq \epsilon \|\E[S]\|
			\]
			holds with probability greater than $ 1-20\exp(-\gamma_\epsilon n)-16/n^2 $, where $ \gamma_\epsilon $ is a positive constant.
			
			For the admissible CDPs model, the same strategy can be applied, and the details are omitted here for brevity. 
		\end{proof}
		\begin{lemma}\label{eiganalysis}
			Given vectors $ \vu\in\C^n $, $ \vv\in\C^n $ and a scalar $ \beta>\|\vu\|^2 $,  define the Hermitian matrix 
			\[
			E = \begin{pmatrix}
				\beta I_n-\vu\vu^*+2\vv\vv^* & 2\vv\vv^\top\\
				2\conj{\vv}\vv^* & \beta I_n-\conj{\vu}\vu^\top+2\conj{\vv}\vv^\top
			\end{pmatrix}.
			\]
			The matrix $ E $ is invertible, with extreme eigenvalues bounded by $ \lambda_{\max}\leq \beta+4\|\vv\|^2 $ and $ \lambda_{\min}=\beta-\|\vu\|^2 $.
		\end{lemma}
		\begin{proof}
			We begin by rewriting the matrix $ E $ as
			\begin{align*}
				E &= \begin{pmatrix}
					\beta I_n-\vu\vu^* & \boldsymbol{0}\\
					\boldsymbol{0} & \beta I_n-\conj{\vu}\vu^\top
				\end{pmatrix}+2\begin{pmatrix}
					\vv\\ \conj{\vv}
				\end{pmatrix}\begin{pmatrix}
					\vv^*& \vv^\top
				\end{pmatrix},
			\end{align*}
			which can be expressed as
			\[
			E:=E_1+2\vs\vs^*,
			\]
			where
			\[
			E_1:=\begin{pmatrix}
				\beta I_n-\vu\vu^* & \boldsymbol{0}\\
				\boldsymbol{0} & \beta I_n-\conj{\vu}\vu^\top
			\end{pmatrix}
			\] and $ \vs^*=\begin{pmatrix}
				\vv^*& \vv^\top
			\end{pmatrix} $.
			
			We claim that $ E $ is invertible if $\beta>\|\vu\|^2 $.  According to Sherman-Morrison formula, $ E $ being invertible is equivalent to 
			\begin{itemize}
				\item $ E_1 $ is invertible,
				\item $ 1+2\vs^*E_1^{-1}\vs\neq 0 $.
			\end{itemize}
			{\bfseries Step 1: Invertibility of $ E_1 $ }
			
			To check if $ E_1 $ is invertible, we first note that $ E_1 $ is block diagonal, so its invertibility is determined by the invertibility of  $\beta I_n -\vu\vu^*  $. 
			Using the Sherman-Marrison formula, we know that $\beta I_n -\vu\vu^*  $ is invertible if and only if $ \beta\neq 0 $ and $ \beta\neq \|\vu\|^2 $. Moreover, when invertible, we have
			\[
			(\beta I_n -\vu\vu^*)^{-1} = \frac{I_n}{\beta}+\frac{\vu\vu^*}{\beta(\beta-\|\vu\|^2)}.
			\]
			Thus, when 
			\begin{align}\label{req1}
				\beta\neq 0\quad\textup{and} \quad \beta\neq \|\vu\|^2,
			\end{align}
			$ E_1 $ is invertible, and its inverse is given by 
			\begin{align}\label{inve1}
				E_1^{-1}=\begin{pmatrix}
					\frac{I_n}{\beta}+\frac{\vu\vu^*}{\beta(\beta-\|\vu\|^2)}&\boldsymbol{0}\\
					\boldsymbol{0}&\frac{I_n}{\beta}+\frac{\conj{\vu}\vu^\top}{\beta(\beta-\|\vu\|^2)}
				\end{pmatrix}.
			\end{align}
			{\bfseries Step 2: Condition for $ 1+2\vs^*E_1^{-1}\vs\neq 0 $ }
			
			We now check the second requirement. Using the definition of $ \vs $ and the expression for $ E_1^{-1}$ (\ref{inve1}), the condition becomes
			\[
			\frac{\|\vv\|^2}{\beta}+\frac{|\vv^*\vu|^2}{\beta(\beta-\|\vu\|^2)}\neq -\frac{1}{4}.
			\]
			Simplifying this, we arrive at the condition 
			\begin{align}\label{req2}
				\beta\neq\frac{1}{2}\|\vu\|^2-2\|\vv\|^2\pm \frac{1}{2}\sqrt{\big(4\|\vv\|^2+\|\vu\|^2\big)^2-16|\vu^*\vv|^2}
			\end{align}
			Furthermore, we observe that 
			\begin{align*}
				\frac{1}{2}\|\vu\|^2-2\|\vv\|^2+ \frac{1}{2}\sqrt{\big(4\|\vv\|^2+\|\vu\|^2\big)^2-16|\vu^*\vv|^2}&\leq \|\vu\|^2,\\
				\frac{1}{2}\|\vu\|^2-2\|\vv\|^2- \frac{1}{2}\sqrt{\big(4\|\vv\|^2+\|\vu\|^2\big)^2-16|\vu^*\vv|^2}&\geq -4\|\vv\|^2.
			\end{align*}
			Thus, when $ \beta>\|\vu\|^2 $, conditions (\ref{req1}) and (\ref{req2}) are satisfied simultaneously, and we can conclude that $ E $ is invertible.
			
			Next, we analyze the eigenvalues of $ E $ by solving the equation $ \det(E-\lambda I_{2n})=0 $. Recall that
			\begin{align*}
				E-\lambda I_{2n} &= \begin{pmatrix}
					(\beta-\lambda )I_n-\vu\vu^* & \boldsymbol{0}\\
					\boldsymbol{0} & (\beta -\lambda)I_n-\conj{\vu}\vu^\top
				\end{pmatrix}+2\begin{pmatrix}
					\vv\\ \conj{\vv}
				\end{pmatrix}\begin{pmatrix}
					\vv^*& \vv^\top
				\end{pmatrix}\\
				&:=E_2+2\vs\vs^*.
			\end{align*}
			Then we only need to find $ \lambda $, under which $ E_2+2\vs\vs^* $ is invertible. 
			Similar to the process we had already calculated, the solutions to $  \det(E-\lambda I_{2n})=0  $ are given by 
			$$ 
			\left\{  
			\begin{aligned}
				\lambda&=\beta,   \\  
				\lambda&=\beta-\|\vu\|^2,  \\  
				\lambda&=\beta-\Big(\frac{1}{2}\|\vu\|^2-2\|\vv\|^2+ \frac{1}{2}\sqrt{\big(4\|\vv\|^2+\|\vu\|^2\big)^2-16|\vv^*\vu|^2}\Big)\geq\beta-\|\vu\|^2 \\
				\lambda&=\beta -\Big(\frac{1}{2}\|\vu\|^2-2\|\vv\|^2-\frac{1}{2}\sqrt{\big(4\|\vv\|^2+\|\vu\|^2\big)^2-16|\vv^*\vu|^2}\Big) \leq\beta+4\|\vv\|^2
			\end{aligned} 
			\right. .
			$$
			Therefore, we conclude that
			\begin{align*}
				\lambda_{\min} &= \beta-\|\vu\|^2, \\
				\lambda_{\max} &\leq \beta+4\|\vv\|^2.
			\end{align*}
		\end{proof}
		{
			\bibliographystyle{plain}
			\bibliography{bib_affine_phase_retrieval}
		}
	\end{document}